\documentclass[reprint,aps,nofootinbib,twocolumn]{revtex4-2}
\usepackage{amsmath}
\usepackage{amsfonts}
\usepackage{amssymb}
\usepackage{amsthm}
\usepackage{mathtools}
\usepackage{graphicx}
\usepackage[dvipsnames]{xcolor}
\usepackage[colorlinks=True,citecolor=gray,linkcolor=myRed,urlcolor=gray,hypertexnames=false]{hyperref}
\usepackage{hypcap}
\usepackage{yfonts}
\usepackage{bm}
\usepackage[normalem]{ulem}
\usepackage{dsfont}
\usepackage{braket}
\usepackage{marginnote}
\usepackage{adjustbox}
\usepackage{tikz}
\usepackage{quantikz}
\usetikzlibrary{shapes, arrows.meta}
\newcommand\cbox[1]{\vcenter{\hbox{#1}}}
\usepackage{kantlipsum}
\newtheorem{theorem}{Theorem}

\newcommand\eq[1]{\begin{align}#1\end{align}}

\newcommand\as[1]{{\color{black}{#1}}}
\newcommand\add[1]{{\color{black}{#1}}}

\newcommand\mytitle{Information phases of partial projected ensembles generated from random quantum states and scrambling dynamics}

\definecolor{myBlue}{RGB}{31,119,180}
\definecolor{myOrange}{RGB}{255,127,14}
\definecolor{myGreen}{RGB}{44,160,44}
\definecolor{myRed}{RGB}{214,39,40}
\definecolor{myPurple}{RGB}{148,103,189}

\begin{document}

\title{\mytitle}

\author{Alan Sherry}
\email{alan.sherry@icts.res.in}
\affiliation{International Centre for Theoretical Sciences, Tata Institute of Fundamental Research, Bengaluru 560089, India}

\author{Saptarshi Mandal}
\email{saptarshi.mandal@icts.res.in}
\affiliation{International Centre for Theoretical Sciences, Tata Institute of Fundamental Research, Bengaluru 560089, India}

\author{Sthitadhi Roy}
\email{sthitadhi.roy@icts.res.in}
\affiliation{International Centre for Theoretical Sciences, Tata Institute of Fundamental Research, Bengaluru 560089, India}

\begin{abstract}
The projected ensemble -- an ensemble of pure states on a subsystem conditioned on projective measurement outcomes on its complement -- provides a finer probe of ergodicity and information structure than the reduced density matrix of the subsystem in bipartite quantum states. This framework can be generalised to partial projected ensembles in tripartite settings, where outcomes from part of the measured subsystem are discarded, leading to ensembles of mixed states. We show that information measures defined for such ensembles, in particular the Holevo information, yield a more detailed characterisation of how quantum information is distributed between subsystems compared to conventional entanglement measures. Using exact analytical results supported by numerical results, we uncover a qualitative change in the scaling of the Holevo information with system size in partial projected ensembles generated by Haar-random states, as the relative sizes of the subsystem are varied. In one phase, the Holevo information decays exponentially with system size, while in the other it grows linearly, thereby defining distinct information phases separated by sharp transitions signalled by non-analyticities in the Holevo information. The exponentially decaying phase rigorously establishes the existence of a measurement-invisible quantum-correlated phase -- a manifestation of many-body information scrambling with no bipartite analogue. We contrast this information-phase diagram with the entanglement-phase structure of tripartite Haar-random states obtained from logarithmic negativity, and show that the Holevo information reveals additional fine structure beyond conventional entanglement measures. Finally, we show that these information phases, as characterised by the Holevo information, emerge in the dynamics of chaotic quantum circuits and discuss the associated timescales.
\end{abstract}

\maketitle

\tableofcontents

\section{Introduction\label{sec:intro}}

Any pure quantum state, $\ket{\Psi}$, of $N$ qubits carries a total quantum information which is proportional to $N$.
This can be made precise via the von Neumann information~\cite{nielsen-chaung-book}, which for a state $\rho$ is defined as 
\eq{
I_{\rm vN}(\rho) = \ln({\rm dim}[\rho]) - {\cal S}_{\rm vN}(\rho) \overset{\rho=\ket{\Psi}\bra{\Psi}}{=}N\ln 2\,,
}
where the second equality follows from the fact that the von Neumann entropy, ${\cal S}_{\rm vN}$ of a pure state is identically zero. 
However, the much more fundamental and pertinent question is how this total quantum information is distributed and shared amongst the $N$ qubits.
Often quantified via entanglement measures, this question lies at the heart of several fundamental questions ranging from thermalisation~\cite{nandkishore2015many,dalessio2016from,nahum2017quantum,garrison2018does} and information scrambling~\cite{hayden2007black,sekino2008fast,liu2014entanglement,liu2014entanglement1,hosur2016chaos,xu2019locality} in quantum statistical mechanics to efficient classical simulability and the utility of the entanglement as a resource~\cite{schuch2008entropy,harrow2017nature,boixo2018characterising,arute2019quantum} in quantum information science.

A ubiquitous diagnostic for this structure of quantum information in a state is the entanglement between a subsystem and the rest of the system, quantified by the von Neumann or R\'enyi entropies of the reduced density matrix of the former~\cite{nielsen-chaung-book}.
Beyond this bipartite setting, the entanglement between two disjoint, not necessarily complementary subsystems is often characterised by the logarithmic negativity~\cite{vidal2002computable,plenio2005logarithmic,eisert2006entanglement}; it is an entanglement measure for the mixed state of the two subsystems obtained by tracing out the remaining part of the system.
While such measures have led to fundamental insights about the information content of quantum many-body states, one may argue that these are somewhat coarse measures of the information structure in the states.
The von Neumann or R\'enyi entropies depend solely on the appropriate reduced density matrices, which are in turn obtained by tracing out the complementary subsystems. Physically, this means that all the information of the complementary subsystem is discarded. While the logarithmic negativity does retain some information about the complementary subsystem, it coarse-grains the full density matrix by retaining only the trace norm of its partial transpose.

However, more recently it has been realised that qualitatively more information can be obtained about the statistical properties of the state by retaining some information about the erstwhile traced out `bath' subsystem. 
This is done by performing projective measurements on the `bath' and constructing an ensemble of states on the original subsystem conditioned on the measurement outcomes. 
The ensemble, dubbed the projected ensemble 
manifestly contains qualitatively more information than the reduced density matrix of the subsystem, as the former describes an
entire distribution of states, whereas the latter is just the
first moment of the said distribution~\cite{choi2023preparing,cotler2023emergent}.
For instance, the projected ensemble allows for the generalisation of the ideas of thermalisation to that of {\it deep thermalisation}~\cite{ho2022exact,ippoliti2022solvablemodelofdeep,cotler2023emergent,choi2023preparing,lucas2023freefermiondeepthm,ippolitu2023dynamical,bhore2023deepthmconstrained,varikuti2024unravelingemergence,chan2024pe,mark2024maximum,manna2025peconserved} wherein not just the reduced density matrix of a local subsystem but the entire ensemble is described by a universal, maximum entropy ensemble. 
From an information-theoretic point of view, this is intimately connected to emergence of quantum designs from single wavefunctions~\cite{choi2023preparing,cotler2023emergent}. 

While much of the work on projected ensembles and deep thermalisation has focused on ensembles of pure states, much more recently the relevance of extending these ideas to setting of ensembles of mixed states has been realised~\cite{yu2025mixedstatedeepthermalization,sherry2025mixedstatesexhibitdeep,mandal2025partialprojectedensemblesspatiotemporal}. 
Such situations can naturally arise if the measurement outcomes are partially lost or there is finite mixedness in the state due to state-preparation and measurement errors. 
Interestingly, it was shown that even in such settings a variant of the maximum entropy ensemble continues to emerge~\cite{sherry2025mixedstatesexhibitdeep}.

In light of the above, two fundamental questions arise which we address in this paper.
First, how to develop an information-theoretic framework which exploits the availability of the entire ensemble of states on a subsystem induced by measurements on a disjoint subsystem, and thereby captures the quantum correlations between the two subsystems beyond what conventional entanglement measures can capture.
Second, is there a qualitative change in the behaviour of the information measure as the sizes of the subsystems are varied in the thermodynamic limit, thereby allowing us to define information phases. The answer to this question also sheds insights onto the scales over which the information resides in the state.

In this work, we show that the Holevo information of an ensemble of states on a subsystem conditioned on the measurement outcomes of a disjoint and non-complementary subsystem leads to a finer characterisation of the information structure of the quantum state, compared to conventional entanglement measures.
This constitutes the first main result of this work.
The ensemble so-constructed for a subsystem $R$ can be viewed as a projected ensemble except the measurement outcomes are retained only for {\it part} $S$ of the complement $\overline{R}$ and the outcomes in the remaining subsystem are discarded; this motivates the nomenclature of {\it partial} projected ensemble. 
Given the construction of the ensemble, the Holevo information turns out to be a natural quantity to study as it is the mutual information of the classical-quantum state built out of the measurement outcomes and the conditional states. 

After establishing the Holevo information as the central quantity of interest, we present a detailed analysis of the same for partial projected ensembles generated from Haar-random states, focusing on the setting where all the subsystems are extensive. 
Remarkably, we find that, depending on the relative sizes of the subsystems the Holevo information shows qualitatively different scaling with the number of qubits $N$: in one regime it decays exponentially in $N$ and in the other it grows linearly with $N$.
More importantly, this qualitative change occurs in a non-analytic fashion, thereby establishing concretely the notion of information phases and genuine transitions between them. 
This constitutes the second main result of the work.
For a part of the parameter space, we derive the phase diagram analytically, which is complemented by numerical results in the analytically inaccessible part of the parameter space. 
Interestingly, we find that these two phases of vanishing and volume-law Holevo information respectively lie at the heart of the measurement-invisible quantum-correlated (MIQC) and measurement-visible quantum-correlated (MVQC) phases discovered earlier~\cite{sherry2025miqc}.
In other words, our results for the Holevo information show the fundamental information structure of the states which underpins the MIQC-MVQC phase transition.

With the information-phase diagram for Haar-random states so established through the lens of the Holevo information, a natural question that arises is whether these phases emerge in physically relevant settings of ergodic dynamics.
We answer this question in the affirmative by showing numerically that precisely the same phase diagram emerges dynamically in chaotic quantum circuits built up of $2$-local gates. 
We also analyse the timescales associated to the emergence of the phases and their dependence on the circuit geometry.

The rest of the paper is organised as follows. In Sec.~\ref{sec:ppe} we set up the basic definitions of the partial projected ensemble and its universal properties from the point of view of deep thermalisation. In particular, we lay down rigorously the sufficient conditions for the ensemble to be described by a universal maximum entropy ensemble. In Sec.~\ref{sec:ppe-info} we discuss in detail the relevance of the Holevo information as a {\it bona fide} information measure for partial projected ensembles and their information-theoretic aspects. The computation of the Holevo information for partial projected ensembles generated from Haar-random states constitutes Sec.~\ref{sec:phase-dia}; in this section using the analytical results for the Holevo information for a part of the parameter space (parameters being the relative sizes of the subsystems) we chart out the information phase diagram analytically and complement the phase diagram in the rest of the parameter space with numerical results. 
In Sec.~\ref{sec:circuits}, we report the emergence of the information phases in chaotic dynamics effected by random unitary circuits.
Finally, we close with a summary and outlook in Sec.~\ref{sec:summary}.

\section{Partial projected ensembles and deep thermalisation\label{sec:ppe}}

We start this section by briefly discussing the idea of (partial) measurement-induced (partial) projected ensembles of states defined on a subsystem~\cite{mandal2025partialprojectedensemblesspatiotemporal} and the universal, limiting forms of such ensembles from the point of view of deep thermalisation~\cite{yu2025mixedstatedeepthermalization,sherry2025mixedstatesexhibitdeep}.

\subsection{Defining the (partial) projected ensemble \label{sec:ppe-defs}}

Consider a system of $N$ qubits that is tripartitioned into subsystems $R$, $E$ and $S$.
We will denote the complement of a subsystem by an overline, for instance, $\overline{R} = E\cup S$.
Let the system be in a pure state denoted by $\ket{\Psi}$.
The projected ensemble (PE) on subsystem $R$ is then defined as the ensemble of pure states on $R$ conditioned on measurement outcomes in $\overline{R}$, where the probability of each state given by the Born probability of the corresponding outcome.
For concreteness, we will consider projective measurements on $\overline{R}$ in some local tensor product basis $\{\ket{o_{\overline{R}}}\}$,
where $\{o_{\overline{R}}\}=\{0,1\}^{|\overline{R}|}$ is a set of bitstring measurement outcomes.
The PE can then be formally defined as 
\eq{
{\cal E}_{{\rm PE}_R} \equiv \{p(o_{\overline{R}}),~\ket{\psi_{R}(o_{\overline{R}})}\}_{o_{\overline{R}}}\,,
\label{eq:pe-def}
}
where the probability is given by
\eq{
p(o_{\overline{R}}) = \braket{\Psi|\Pi_{o_{\overline{R}}}|\Psi}\,;\quad \Pi_{o_{\overline{R}}}=\ket{o_{\overline{R}}}\bra{o_{\overline{R}}}\,,
\label{eq:p-orbar}
}
and the conditional state by
\eq{
\frac{\Pi_{o_{\overline{R}}}\ket{\Psi}}{\sqrt{p(o_{\overline{R}})}} = \ket{\psi_R(o_{\overline{R}})}\otimes \ket{o_{\overline{R}}}\,.
\label{eq:psi-orbar}
}
With the PE so defined, its $k^{\rm th}$ moment is defined as 
\eq{
    \rho^{(k)}_{{\rm{PE}}_R}=\sum_{o_{\overline{R}}}p(o_{\overline{R}})\left[\ket{\psi_{R}(o_{\overline{R}})}\bra{\psi_{R}(o_{\overline{R}})}\right]^{\otimes k}\,,
    \label{eq:mom-def-pe}
}
where it is straightforward to see that the first moment ($k=1$) reduces simply to the reduced density matrix of $R$, $\rho^{(1)}_{{\rm{PE}}_R}=\rho_R = {\rm Tr}_{\overline{R}}\ket{\Psi}\bra{\Psi}$.
The higher moments, $k\ge 2$, encode non-linear functionals of the states in the PE and therefore contain information that cannot be reconstructed from $\rho_R$ alone.

Note that the PE is generated over measurement outcomes on the entire subsystem $\overline{R}=E\cup S$. 
However, a more general setting is that the measurement outcomes in a part of $\overline{R}$, say $E$, are lost and the resulting ensemble is over the measurement outcomes only in $S$. This is equivalent to constructing a PE on $R\cup E$ over the measurement outcomes in $S$ and then tracing out $E$ from each member of the ensemble.
This naturally leads to an ensemble of mixed states in $R$ over measurement outcomes only on a part of $\overline{R}$ and hence the name {\it partial projected ensemble} (PPE). 
Formally the PPE is defined as
\eq{
{\cal E}_{{\rm PPE}_R} \equiv \{p(o_S),~\rho_R(o_S)\}_{o_S}\,,
\label{eq:ppe-def}
}
where similar to Eq.~\ref{eq:p-orbar}, the probability is given by 
\eq{
p(o_S) = \bra{\Psi}\Pi_{o_S}\ket{\Psi}\,,
\label{eq:p-os}
}
and the conditional state, now mixed, by
\eq{
\rho_R(o_S) = \frac{{\rm Tr}_{ES}[\Pi_{o_S}\ket{\Psi}\bra{\Psi}\Pi_{o_S}]}{p(o_S)}\,.
\label{eq:rhoR-os}
}
As with the PE (see Eq.~\ref{eq:mom-def-pe}), the moments of the PPE are defined as 
\eq{
    \rho^{(k)}_{{\rm{PPE}}_R}=\sum_{o_S}p(o_S)\left[\rho_R(o_S)\right]^{\otimes k}\,,
    \label{eq:mom-def-ppe}
}
where again $k=1$ trivially corresponds to $\rho_R$ and moments with $k\ge 2$ are the non-trivial moments. 

The interpretation of the PPE on $R$ as the partial trace over $E$ of the PE on $R\cup E$ leads to a simple relation between the moments of the PE on $R\cup E$ and the moments of the PPE on $R$.
The relation is given by
\eq{
\rho_{{\rm{PPE}}_R}^{(k)}={\rm{Tr}}_{{\cal H}_E^{\otimes k}}\left[\rho^{(k)}_{{\rm{PE}}_{R\cup E}}\right]\,,
\label{eq:pe-ppe-mom}
}
where the notation ${\cal H}_E^{\otimes k}$ means the $k$-fold replicated Hilbert space of $E$.
The above relation will be particularly useful to relate the universal, statistical properties of the PE and the PPE as we discuss next.

\subsection{Deep Thermalisation\label{sec:ppe-dt}}

Given that the PEs and the PPEs, by definition, allow access to entire ensembles of states on a subsystem instead of just its RDM, they lead to a fundamental generalisation of the idea of thermalisation under the umbrella of {\it deep thermalisation}~\cite{ho2022exact,ippoliti2022solvablemodelofdeep,cotler2023emergent,choi2023preparing,lucas2023freefermiondeepthm,ippolitu2023dynamical,bhore2023deepthmconstrained,varikuti2024unravelingemergence,chan2024pe,mark2024maximum,manna2025peconserved,bejan2025matchgate,yu2025mixedstatedeepthermalization,sherry2025mixedstatesexhibitdeep}.

The statement of thermalisation is that for a sufficiently ergodic state, such as those obtained from time-evolution with chaotic Hamiltonians, the RDM of a local subsystem $R$ is well described by a Gibbs state,
$\rho_R \sim e^{-\beta H_R}$,
where the temperature $\beta^{-1}$ is set by the average energy of the initial state with respect to the Hamiltonian~\cite{deutsch1991quantum,srednicki1994chaos,rigol2008thermalisation,deutsch2018eigenstate}.
For dynamics which do not conserve energy, such as those effected by Floquet circuits or random unitary circuits, the system heats up to infinite temperature such that $\rho_R\sim \mathbb{I}$~\cite{lazarides2014equilibrium,dalessio2014long}. These limiting thermodynamic states are universal in the sense that they do not depend on the details of the initial state except for macroscopic properties such as energy.

Deep thermalisation extends this idea to the entire ensemble and makes a statement about the form of the universal, limiting ensembles that the PE or the PPE tend towards. 
Specifically, it states that an ensemble is a deep thermal ensemble, if it maximises the ensemble entropy~\cite{mark2024maximum} subject to the constraints or conservation laws it must satisfy. 
In other words, the states in the ensemble must cover the available Hilbert space in a maximally ergodic fashion.

In the absence of conservation laws, a deep thermal PE is described by the Haar ensemble, which, by definition, is the maximum entropy ensemble on a given Hilbert space. For a subsystem $R$ with Hilbert-space dimension $D_R$, the $k^{\rm th}$ moment of the ensemble is given by
\eq{
\rho_{{\rm Haar}_R}^{(k)}=\int_{\psi\sim{\rm Haar}(D_R)}d\psi\left(\ket{\psi}\bra{\psi}\right)^{\otimes k}\,.
\label{eq:haar-mom}
}
Conservation laws modify the deep thermal ensemble by distorting the Haar ensemble to generate the so-called Scrooge ensemble if the measurement basis is orthogonal to the conserved charges. 
More generally, if the RDM is constrained to a specific $\rho_R$, the corresponding Scrooge ensemble is given by~\cite{jozsa1994accessinfo,goldstein2006distribution,lucas2023freefermiondeepthm,chang2024gaussian,chang2025charge,manna2025peconserved,bejan2025matchgate}
\eq{
{\cal E}_{{\rm Scr}_R}[\rho_R] = \left\{D_R\braket{\psi|\rho_R|\psi},~\frac{\sqrt{\rho}\ket{\psi}}{\sqrt{\braket{\psi|\rho_R|\psi}}}\right\}_{\psi\sim {\rm Haar}_R}\,.
\label{eq:scrooge-def}
}
From an information-theoretic point of view, the Scrooge ensemble minimises the accessible information~\cite{jozsa1994accessinfo}, which quantifies the maximum amount of information that can be gained from the measurement of a state from these ensembles~\cite{holevo1973accinfo,preskill2025quantumshannontheory}.

Given these universal forms of the PE, the corresponding universal forms for the PPE follow straightforwardly from the interpretation of the PPE as a partial trace over the PE defined on an augmented system (see, for instance, Eq.~\ref{eq:pe-ppe-mom}).
In the absence of any conservation laws, the PPE is described by the generalised Hilbert-Schmidt ensemble (gHSe), which is simply the partial trace over the states in a Haar ensemble on the augmented system~\cite{hall1998randomquantumcorrelations,karol2001inducedmeasures,bansal2025pseudordm}. More precisely, the gHSe on subsystem $R$ is defined as 
\eq{
{\cal E}_{{\rm gHSe}_R}\equiv\{\mathrm{Tr}_E\left[\ket{\psi}\bra{\psi}\right]:\psi\sim {\rm Haar}_{R\cup E}\}\,.
\label{eq:ghse}}
Note that the gHSe is uniquely specified by just the Hilbert-space dimensions, $D_R$ and $D_E$ of subsystems $R$ and $E$ respectively. 
The $k^{\rm th}$ moment of the gHSe can be expressed as 
\eq{
\rho_{{\rm gHSe}_R}^{(k)} = {\rm Tr}_{{\cal H}_E^{\otimes k}}\left[\rho_{{\rm Haar}_{R\cup E}}^{(k)}\right]\,.
\label{eq:ghse-mom}
}

More generally, if the PPE on $R$ is constrained to a specific first moment $\rho_R$, then the maximum entropy ensemble for the PPE is given by the partial trace over a Scrooge ensemble of pure states (Eq.~\ref{eq:scrooge-def}) defined on an augmented system which is nothing but a purification of $\rho_R$~\cite{sherry2025mixedstatesexhibitdeep}.
However, in this work we will deal with random unitary circuits with no conservation laws, and hence, it is the gHSe in Eq.~\ref{eq:ghse} which is of relevance.

\as{To quantify the distance of the PPE from the gHSe, we will use the ubiquitously used measure of the trace norm of the difference between their $k^{\rm th}$ moments,
\eq{
\Delta_{{\rm gHSe}_R}^{(k)} = \frac{1}{2}\left|\left|\rho_{{\rm PPE}_R}^{(k)} - \rho_{{\rm gHSe}_R}^{(k)}\right|\right|_1\,,
\label{eq:Delta-gHSe}
}
where the trace norm of an operator $O$ is defined as $\Vert O\Vert_1={\rm{Tr}}\left[\sqrt{O^\dagger O}\right]$.} The relation in Eq.~\ref{eq:pe-ppe-mom} leads to $\Delta_{{\rm gHSe}_R}^{(k)}$ being bounded from above by the distance between the PE on $R\cup E$ from a Haar ensemble. 
To see this, we define
\eq{
\Delta_{{\rm Haar}_{R\cup E}}^{(k)} = \frac{1}{2}\left|\left|\rho_{{\rm PE}_{R\cup E}}^{(k)} - \rho_{{\rm Haar}_{R\cup E}}^{(k)}\right|\right|_1\,.
\label{eq:Delta-haar}
}
In addition, note that the partial trace, which is a completely positive trace-preserving map, is contractive under the trace norm~\cite{nielsen-chaung-book}.
That implies, given a state $\rho\in{\cal H}_R\otimes{\cal H}_E$ 
\eq{
\left\Vert\rm{Tr}_{{\cal H}_E}\left[\rho\right]\right\Vert_1\leq \Vert\rho\Vert_1\,.\label{eq:td-contract}
}
Using Eq.~\ref{eq:pe-ppe-mom} in the inequality in Eq.~\ref{eq:td-contract} immediately leads to the inequality
\eq{
\Delta_{{\rm gHSe}_R}^{(k)}\leq \Delta_{{\rm Haar}_{R\cup E}}^{(k)} ~\forall k\,.\label{eq:ghse-haar}
}
The bound in Eq.~\ref{eq:ghse-haar} ensures that if the PE on $R\cup E$ forms an $\epsilon$-approximate $k$-design, the distance between the respective $n^{\rm{th}}$ moments of the corresponding PPE on $R$ and the appropriate gHSe are bounded from above by $\epsilon$ for any $n\leq k$.

\subsection{Haar random states and extensive subsystems}
At this stage, it is useful to note that much of the work on projected ensembles and deep thermalisation in the literature has focused on the case where $R$ and $E$ are local subsystems and the measured `bath' $S$ is taken to be thermodynamically large~\cite{ho2022exact,ippoliti2022solvablemodelofdeep,lucas2023freefermiondeepthm,ippolitu2023dynamical,bhore2023deepthmconstrained,varikuti2024unravelingemergence,chan2024pe,mark2024maximum,manna2025peconserved,bejan2025matchgate,sherry2025mixedstatesexhibitdeep}. 
The case where all three subsystems are extensive is relatively less understood. 
However, this is exactly the setting that is of concern in this work, as mentioned in Sec.~\ref{sec:intro}.
We therefore present some results for the emergence of maximum entropy ensembles for the PE as well as the PPE in such a scenario.
More specifically, given that our interest lies in scrambling dynamics without any conservation laws, a natural setting for us is when the state $\ket{\Psi}$ is a Haar random state. 
In the following, we lay down rigorously the sufficient conditions for the emergence of the Haar ensemble for the PE and the gHSe for the PPE for extensive $R$, $E$ and $S$.

\subsubsection{Emergence of Haar ensemble for the PE}

For a Haar random state $\ket{\Psi}$, the first moment, $\rho_R$, of the PE on $R$, is known to be maximally mixed for $|R|\le N/2$ as $N\to\infty$~\cite{page1993average,forrestor2010rmt}. For the higher moments, the condition is more subtle~\cite{cotler2023emergent}.
A theorem in Ref.~\cite{cotler2023emergent}
states that if   
\eq{
D_{\overline{R}}\geq \frac{18\pi^3(2k-1)(D_R)^{4k}}{\epsilon^2}(2k\ln(D_R)+\ln(2/\delta))\,,
\label{eq:cotler-haar-bound}
}
then the PE on $R$ forms an $\epsilon$-approximate $k$-design with probability at least $1-\delta$. 

Considering $R$ to be extensive, such that $|R| = \gamma N$ for some $\gamma\in[0,1]$ and $|\overline{R}| = (1-\gamma)N$, the asymptotic form of the bound in Eq.~\ref{eq:cotler-haar-bound} as $N\to\infty$ is given by
\eq{
\gamma\leq \frac{1}{4k+1}\,.
}
If the above bound on the dimensions of the subsystem $R$ is satisfied, then the PE on $R$ is asymptotically an exact $k$-design. 
This lays down a rigorous condition for which the $n^{\rm th}$ moments of the PE, for $n\leq k$, on an extensive $R$ approach that of the Haar ensemble.

\subsubsection{Emergence of gHSe for the PPE}
In the case of tripartite states, the conditions for the emergence of a gHSe on $R$ induced by measurements on $S$ and tracing out $E$ are somewhat more complicated.
First note that, if the PE $R\cup E$ forms a Haar ensemble, then the PPE on $R$ trivially forms a gHSe. 
Using the inequality in Eq.~\ref{eq:cotler-haar-bound}, we therefore have
\eq{
D_{S}\geq \frac{18\pi^3(2k-1)(D_RD_E)^{4k}}{\epsilon^2}\left(\ln\frac{2(D_RD_E)^{2k}}{\delta}\right),
}
as a sufficient condition for PPE on $R$ to form an $\epsilon$-approximate $k$-gHSe\footnote{We will refer to a PPE as an $\epsilon$-approximate $k$-gHSe if the corresponding $\Delta_{{\rm gHSe}}^{(n)}<\epsilon$ for all $n\leq k$.} with probability at least $1-\delta$. 
However, this bound is suboptimal for the PPE \as{since it is derived from the properties of the PE, which is defined over a larger subsystem $R\cup E$.}
In the following we provide rigorously an improvement to the bound.
This can be stated in the form of the following theorem which constitutes the first result of this work (see Ref.~\cite{yu2025mixedstatedeepthermalization} for an analogous theorem).
\begin{theorem}
\label{thm:ghse}
Let $\ket{\Psi}$ be a tripartite state chosen uniformly at random from a $D_R\times D_S\times D_E$ dimensional Hilbert space. The moments of ${\cal E}_{{\rm PPE}_R}$  satisfy $\Vert \rho^{(k)}_{{\rm{PPE}}_R}-\rho^{(k)}_{{\rm{gHSe}}_R}\Vert_1\leq \epsilon$ with probability at least $1-\delta$ if
	\eq{D_SD_E\geq\frac{18\pi^3(2k-1)^2D_R^{4k-1}}{\epsilon^2}\ln\left(\frac{2 D_R^{2k}}{\delta}\right)\,.\label{eq:ghse-bound}}	
\end{theorem}

While we present the proof of theorem in detail in Appendix~\ref{app:ghse-thm}, it is more useful to analyse the large-$N$ limit of the bound here.
To analyse this limit, let us assume that subsystems $R$, $S$ and $E$ are extensive such that $|R|=\gamma N$, $|S|= pN$ and $|E|=(1-p-\gamma)N$ for some $p\in(0,1)$ and $\gamma\in(0,1-p)$. 
In addition, we assume that the bound on the distance $\epsilon$ asymptotically decays in $N$ as $2^{-\alpha N}$, for some $\alpha>0$. 
This gives the asymptotic form of the bound in Eq.~\ref{eq:ghse-bound} as $N\rightarrow\infty$ as
\eq{\gamma\leq \frac{1-2\alpha}{4k}\,.\label{eq:asym-bound-ghse}}
If this bound on the size of subsystem $R$ is satisfied, then the first $k$ moments of the PPE converge to those of the gHSe no slower than $2^{-\alpha N}$ as $N\rightarrow\infty$. 

In the rest of this work, we will be interested only in the second moment of the PPE. Given a particular $\gamma$, Eq.~\ref{eq:asym-bound-ghse} can be restated as an upper bound on the decay rate in $N$, $\alpha$, of $\epsilon$ for $k=2$, 
\eq{\alpha\leq \frac{1-8\gamma}{2}\,,\label{eq:ghse-optimal}} 
where $\left\Vert\rho^{(2)}_\text{PPE}-\rho^{(2)}_\text{gHSe}\right\Vert_1\leq \epsilon$.
The exponential convergence of the second moment of the PPE to that of the gHSe will play a key role in quantifying the information content of partial projected ensembles, which we will now discuss.

\section{Information measures for partial projected ensembles\label{sec:ppe-info}}

In this section, we present the main framework of this work, namely concrete quantification of the information content of a PPE induced on the subsystem $R$ by measurements on a disjoint subsystem $S$.
It is therefore clear that information measures for the PPE must quantify the correlations between $R$ and $S$.
The simplest such measure arguably is a measure of entanglement between $R$ and $S$.
Since the state $\rho_{RS}={\rm Tr}_E\ket{\Psi}\bra{\Psi}$ is generally a mixed state, the entanglement between $R$ and $S$ can be quantified via the logarithmic negativity~\cite{plenio2005logarithmic,eisert2006entanglement}
\eq{{\cal N}_{RS}\equiv \ln \left\Vert \rho_{RS}^{{\rm T}_S}\right\Vert_1\,,\label{eq:log-neg}}
where ${\rm T}_S$ denotes the partial transpose of $\rho_{RS}$ over $S$.
However, since the PPE is an entire ensemble of states in $R$ over measurement outcomes in $S$, it manifestly holds the potential to offer finer probes of the information content than an entanglement measure.
This was already hinted at by the discovery of the measurement-invisible quantum-correlated (MIQC) phase~\cite{sherry2025miqc} wherein each state in the PPE over $R$ was identical despite there being finite entanglement between $R$ and $S$ indicated by ${\cal N}_{RS}\neq 0$; in other words the measurement outcome $o_S$ had no bearing on the conditional state, $\rho_R(o_S)=\rho_R ~\forall~o_S$.
This already suggests that the statistical properties of the PPE indeed contain finer probes of the structure of how quantum information is encoded across the two subsystems $R$ and $S$.

Given that we have the entire ensemble of states, as in Eq.~\ref{eq:ppe-def}, and not just the RDM $\rho_{RS}$, the question that arises therefore is what measure of information can, in fact, exploit it. Since the PPE is effectively a collection of classical bitstrings $o_S$ in $S$ and conditional quantum states $\rho_R(o_S)$ in $R$, an appropriate measure is the mutual information between $R$ and $S$ in the classical-quantum state~\cite{Henderson2001CQCorr} constructed from the PPE
\eq{
\rho_{\rm CQ} = \sum_{o_S}p(o_S) \rho_R(o_S) \otimes \ket{o_S}\bra{o_S}\,.
\label{eq:rho-cq}
}
The mutual information between $R$ and $S$ in the state Eq.~\ref{eq:rho-cq} \as{is given by}
\eq{
\chi({\cal E}_{{\rm PPE}_R})={\cal S}_{\rm{vN}}\left[\rho_R\right]-\sum_{o_S}p(o_S){\cal S}_{\rm{vN}}\left[\rho_R(o_S)\right]\,,
\label{eq:mi-cq-state}
}
where ${\cal S}_{\rm{vN}}[\rho]$ is the von Neumann entropy of $\rho$. 
Most importantly, note that the mutual information in Eq.~\ref{eq:mi-cq-state} is nothing but the Holevo information~\cite{holevo1973accinfo} of the PPE.

This ties in very well with the key insight that the distribution of the state $\rho_R(o_S)$ in the PPE over $o_S$ carries fundamental signatures of the information content of the PPE. 
The Holevo information of an ensemble can also be interpreted as a measure of how distinct is the first moment ($\rho_R$) of the ensemble and any particular state $\rho_R(o_S)$ in ${\cal E}_{{\rm PPE}_R}$. To see this, note that a measure of distinguishability of two quantum states is the quantum relative entropy of one with respect to the other. The distance between $\rho_R(o_S)$ and $\rho_R$ can therefore be quantified by
\eq{
{\cal S}_{\rm{vN}}\left[\rho_R(o_S)\Vert\rho_R\right]\equiv
-{\cal S}_{\rm{vN}}[\rho_R(o_S)] - {\rm Tr}[\rho_R(o_S)\ln \rho_R]\,.\label{eq:rel-ent}
}
Averaging Eq.~\ref{eq:rel-ent} over ${\cal E}_{{\rm PPE}_R}$ gives
\eq{
\overline{{\cal S}_{\rm{vN}}\left[\rho_R(o_S)\Vert\rho_R\right]}&=\sum_{o_S}p(o_S){\cal S}_{\rm{vN}}\left[\rho_R(o_S)\Vert\rho_R\right] \nonumber\\&= {\cal S}_{\rm{vN}}\left[\rho_R\right]-\sum_{o_S}p(o_S){\cal S}_{\rm{vN}}\left[\rho_R(o_S)\right]\,,
}
which is nothing but the Holevo information, defined in Eq.~\ref{eq:mi-cq-state}, of the PPE, ${\cal E}_{{\rm PPE}_R}$.

The above discussion establishes the Holevo information of ${\cal E}_{{\rm PPE}_R}$ as a fundamental framework and quantity of interest from the point of view of probing the information structure of the PPE. Much of the rest of the paper will be about concrete results for $\chi({\cal E}_{{\rm PPE}_R})$ in pertinent settings. However, before delving into that it will be useful to discuss concretely the information-theoretic aspects of our framework. 

The question of quantifying the information in the PPE can be formulated as a task where two parties, Alice and Bob, communicate via a quantum channel. Consider a task where Alice transmits a classical message
$O_S\equiv \{o_{S_1}, o_{S_2}, \cdots, o_{S_n}\}$ to Bob by sequentially sending $n$ states $\{\rho_R(o_{S_1}),\rho_R(o_{S_2}),\cdots,\rho_R(o_{S_n})\}$ sampled from ${\cal E}_{{\rm PPE}_R}$. For now, we assume that Alice and Bob are communicating via a noiseless channel. Bob uses a simple measurement strategy to decode this message: measure each state separately with some \as{positive operator-valued measure (POVM)}, $O_R$, with elements $\{o_R\}$. The information gained by Bob about Alice's message after each measurement is quantified by the mutual information $I(O_S:O_R)$, defined as
\begin{equation}
    \begin{split}
        I(O_S:O_R)=&-\sum_{o_S}p(o_S)\ln[p(o_S)]\\
&+\sum_{o_R}p(o_R)\sum_{o_S}p(o_S|o_R)\ln[p(o_S|o_R)]\,.
    \end{split}
\end{equation}
where $p(o_R)$ is the distribution of the outcomes of the measurement $O_R$. Given Bob's measurement strategy, the maximum classical transmission rate is given by the \textit{accessible information}
\eq{{\rm{Acc}}({\cal E}_{{\rm PPE}_R}) \equiv \max_{O_R} I(O_S:O_R)\,,\label{eq:acc-info}}
which is simply the mutual information maximised over the measurement $O_R$. 
Crucially, the accessible information is bounded above by the Holevo information,
\eq{\rm{Acc}({\cal E}_{{\rm PPE}_R})\leq \chi({\cal E}_{{\rm PPE}_R})\,.}
This bound is saturated for any ensemble consisting of density matrices that mutually commute~\cite{Holevo1998capacity}. For a generic ensemble, however, Bob’s optimal information gain as quantified by the accessible information falls short of $\chi({\cal E}_{{\rm PPE}_R})$. But given that Alice and Bob use the channel multiple times, there is a possibility of leveraging some type of coding and collective measurement strategy to enhance this information gain. Consider the following: instead of measuring each state sent by Alice separately, Bob decides to decode Alice's message by performing an (optimal) collective measurement that would project Alice's states onto a basis of entangled states. Alice then decides to choose inputs $\{\rho_R(o_{S_1}),\rho_R(o_{S_2}),\cdots,\rho_R(o_{S_n})\}$ that would be more distinguishable under Bob's measurement strategy. She achieves this by choosing from a joint ensemble of $n$-fold input states that is different from ${\cal E}_{{\rm PPE}_R}^{\otimes n}$~\footnote{Note that the joint ensemble of $n$-fold input states is chosen such that the marginal 1-fold ensemble remains ${\cal E}$.}. 
Alice's coding strategy in conjunction with Bob's collective measurement strategy can achieve an information gain of $\chi({\cal E})$ bits per channel use asymptotically as $n\rightarrow\infty$~\cite{schumacher1997classical,Holevo1998capacity,preskill2025quantumshannontheory}. Thus, for an appropriately defined operation, the Holevo information is an operationally relevant measure of the information that is extractable from an ensemble of quantum states.

\section{Holevo information for Haar-random states\label{sec:phase-dia}}

Having established the Holevo information as a \textit{bona-fide} measure of the information content of an ensemble of quantum states, we now turn to analysing its behaviour for the PPE (defined as ${\cal E}_{{\rm PPE}_R}$ in Eq.~\ref{eq:ppe-def}) generated from a 
Haar-random
state $\ket{\Psi}$ over $N$ qubits. 
To recall the setting, we consider a tripartition of the system with subsystem $R$ containing $\gamma N$ qubits, the measured subsystem $S$ containing $p N$ qubits, and the remaining $(1-p-\gamma)N$ qubits in subsystem $E$ being traced out.
Specifically, we will analytically compute the Holevo information (in Eq.~\ref{eq:mi-cq-state}) of 
${\cal E}_{{\rm PPE}_R}$ as a function of $p$ and $\gamma$ in the large-$N$ limit, and show rigorously that depending on the value of $\gamma$ and $p$, $\chi({\cal E}_{{\rm PPE}_R})$ shows  qualitatively different behaviour -- indicating different information phases of the ensemble.

In the following, we discuss the computation of $\chi({\cal E}_{{\rm PPE}_R})$ for a Haar-random state in Sec.~\ref{sec:comp-chi} and in Sec.~\ref{sec:pd-haar}, we discuss the phase diagram in the parameter space of $p$-$\gamma$ and more importantly, contrast it with the information content as encoded in the logarithmic negativity.
Specifically, we will derive the phase diagram analytically for a part of the $p$-$\gamma$ parameter space and provide numerical results for the complementary parameter space.

\subsection{Computation of $\chi({\cal E}_{{\rm PPE}_R})$ \label{sec:comp-chi}}

Since $\chi({\cal E}_{{\rm PPE}_R})$ in Eq.~\ref{eq:mi-cq-state} is constructed explicitly out of the von Neumann entropies of the member states, $\rho_R(o_S)$ of the PPE, it can be naturally expressed in terms of the eigenvalues of $\rho_R(o_S)$ which we will denote as $\{\lambda_{i,o_S}\}$.
To compute $\chi({\cal E}_{{\rm PPE}_R})$, we will therefore first prove a result pertaining to the statistics of $\{\lambda_{i,o_S}\}$, which in turn will use the result that the moments of the PPE converge to those of the gHSe, sufficient conditions for which were laid out in Theorem~\ref{thm:ghse}.

Specifically, we will concern ourselves with the distribution of the non-zero eigenvalues of $\rho_R(o_S)$ over both, the eigenvalues as well as the states in the ensemble. 
Formally, it is defined as 
\eq{
P(\lambda)\equiv\frac{1}{{\cal N}}\sum_{o_S}p(o_S)\sum_{i=1}^{r(o_S)}\delta(\lambda-\lambda_{i,o_S})\,,
\label{eq:eig-density}
}
where $\{\lambda_{i,o_S}\}$ is the set of non-zero eigenvalues of $\rho_R(o_S)$ and $r(o_S)$ is its rank. The normalisation $\mathcal{N}$ is equal to the average rank of the states in ${\cal E}_{{\rm PPE}_R}$. The distribution $P(\lambda)$ determines the Holevo information, since the latter depends only on the spectrum of each state $\rho_R(o_S)$ in the ensemble. Remarkably, we find that the assumption that the second moment of the PPE converge to that of the gHSe (as mandated by Theorem~\ref{thm:ghse}) strongly constrains the form of $P(\lambda)$ in the large-$N$ limit,
which then has implications for $\chi({\cal E}_{{\rm PPE}_R})$.

For any given system size $N$, the approximate emergence of the first two moments of the gHSe in the PPE implies that $P(\lambda)$ satisfies a concentration inequality. This can be stated in the form of the following theorem.
\begin{theorem}
\label{thm:Plambda}
Let $\left\Vert\rho^{(2)}_{{\rm PPE}_R}-\rho^{(2)}_{\rm{gHSe}_R}\right\Vert_1\leq\epsilon$. Assuming that the average rank of the states in ${\cal E}_{{\rm PPE}_R}$ is $\text{min}(D_R,D_E)$, the distribution of non-zero eigenvalues $P(\lambda)$ of the states in the ${\cal E}_{{\rm PPE}_R}$ satisfies the following concentration inequality
\eq{ 
\begin{split}
    \text{Prob}\left[\left\vert \lambda-\mu\right\vert\geq\delta\right]\leq \frac{\mu}{\delta^2}\left[\frac{D_E+D_R}{1+D_RD_E}+\epsilon+\mu\right]\,,
\label{eq:thm-2}
\end{split}
}
where $\mu=\max\{D_E^{-1},D_R^{-1}\}$ is the mean of the distribution. 
\end{theorem}

\begin{proof}
From the definition of $P(\lambda)$ in Eq.~\ref{eq:eig-density}, we have 
\begin{equation}
\sum_{o_S}p(o_S)~\text{Tr}\left[\rho^k_R(o_S)\right]=\mathcal{N}\int d\lambda~ P(\lambda)~\lambda^k\,,
\label{eq:dens-trace}
\end{equation} 
where $k$ is any positive integer and ${\cal N}=\min\{D_R,D_E\}$. While this valid for any positive integer $k$, we restrict ourselves to the case where $k=1,2$. For $k=1$, we obtain
\eq{\int d\lambda~ P(\lambda)~\lambda=\frac{1}{\cal N}\,,}
which states that the mean $\mu$ is simply ${\cal N}^{-1}=\max\{D_R^{-1},D_E^{-1}\}$.

For $k=2$, we rewrite the LHS of Eq.~\ref{eq:dens-trace} by applying the following identity
\eq{\sum_{o_S}p(o_S)~\text{Tr}\left[\rho^2_R(o_S)\right]=\text{Tr}\left({\mathbb{S}}_{D_R^2}\rho^{(2)}_{{\rm PPE}_R}\right)\,,
\label{eq:perm-id}}
where ${\mathbb{S}}_{D_R^2}$ is the $\mathsf{SWAP}$ operator acting on the doubled Hilbert space of $R$. From Hölder's inequality, we have
\begin{equation} 
\begin{split}
\left|\text{Tr}\left[{\mathbb{S}}_{D_R^2}\left(\rho^{(2)}_{{\rm PPE}_R}-\rho^{(2)}_{\rm{gHSe}_R}\right)\right]\right|&\leq \left\Vert\rho^{(2)}_{{\rm PPE}_R}-\rho^{(2)}_{\rm{gHSe}_R}\right\Vert_1\left\Vert
{\mathbb{S}}_{D_R^2}\right\Vert_\infty\\
&\leq \epsilon\,,
\end{split}
\end{equation}
where we used $\left\Vert
{\mathbb{S}}_{D_R^2}\right\Vert_\infty=1$. This inequality enables us to rewrite Eq.~\eqref{eq:perm-id} as
\begin{equation}
	\mathcal{N}\int d\lambda~ P(\lambda)~\lambda^2=\text{Tr}\left({\mathbb{S}}_{D_R^2}\rho^{(2)}_{\rm{gHSe}_R}\right)+\eta\,,
    \label{eq:sec-mom-id}
\end{equation}
where 
\eq{|\eta|\leq \epsilon \,.\label{eq:eta-bnd}} 
In Eq.~\ref{eq:sec-mom-id}, we use the known form of the second moment for the gHSe~\cite{karol2001inducedmeasures}
\begin{equation}
	\rho^{(2)}_{{\rm gHSe}_R}=\frac{D_E^2 ~{\mathbb{I}}_{D_R^2}+D_E~{\mathbb{S}}_{D_R^2}}{D_RD_E(D_RD_E+1)}\,,
\end{equation}
which gives
\begin{equation}
	\begin{split}
\mathcal{N}\int d\lambda~ P(\lambda)~\lambda^2&=\frac{D_E+D_R}{D_RD_E+1}+\eta\,.
\label{eq:dens-val-2}
	\end{split}
\end{equation}
With the first two moments of $P(\lambda)$ thus obtained, we can compute the variance, which enables us to employ the following bound given by Chebyshev's inequality: 
\begin{equation}
\begin{split}
     \text{Prob}\left[\left\vert \lambda-\mu\right\vert\geq\delta\right]&\leq \frac{1}{{\cal N}\delta^2}\left[\frac{D_E+D_R}{1+D_RD_E}+\eta-\frac{1}{{\cal N}}\right]\\
     &\leq \frac{\mu}{\delta^2}\left[\frac{D_E+D_R}{1+D_RD_E}+\epsilon+\mu\right]
\end{split}
\end{equation}
which is the inequality in Eq.~\ref{eq:thm-2}. The latter bound comes from the absolute values of individual terms of the RHS in the former bound, and Eq.~\ref{eq:eta-bnd}.
\end{proof}

We now apply Theorem~\ref{thm:Plambda} to ${\cal E}_{{\rm PPE}_R}$ generated from Haar-random states, where the assumption that the average rank of the states in ${\cal E}_{{\rm PPE}_R}$ is ${\rm min}(D_R,D_E)$ typically holds true. In particular, we focus on how the bound in Eq.~\ref{eq:thm-2} constrains  the form of $P(\lambda)$ in the large-$N$ limit. For technical simplicity, we make the following assumption, which restricts our analysis to a particular subregion of the $p-\gamma$ parameter space. 
We assume that \eq{\epsilon \ll\frac{D_E+D_R}{1+D_ED_R}\,,\label{eq:assum1}} 
for $N\gg1$, which ensures that the leading order term in the RHS of Eq.~\ref{eq:thm-2} is $\frac{\mu(D_E+D_R)}{\delta^2(1+D_ED_R)}$. In the large-$N$ limit, Eq.~\ref{eq:assum1} can be restated in terms of the following bound:
\eq{\alpha>\max\{\gamma,1-p-\gamma\}\,,\label{eq:assum1A}}
where $\alpha$ is the decay rate of $\epsilon$ in the large-$N$ limit. 
For ${\cal E}_{{\rm PPE}_R}$ generated from a Haar-random state, we have the bound on $\alpha$ from Eq.~\ref{eq:ghse-optimal},  
\eq{\alpha\leq\frac{1-8\gamma}{2}\,,}
for the convergence of the first two moments of ${\cal E}_{{\rm PPE}_R}$ to the gHSe.
This in turn implies that for 
\eq{\max\{\gamma,1-p-\gamma\}<\frac{1-8\gamma}{2}\,,\label{eq:assum-main}}
Eq.~\ref{eq:assum1A} is always valid. 

Thus, given a $(p,\gamma)$ that satisfies the bound in Eq.~\ref{eq:assum-main}, and for any choice of $\delta$ such that $\lim_{N\to\infty}\delta \to 0$, and 
\eq{
\lim_{N\rightarrow\infty}&\frac{\mu(D_E+D_R)}{\delta^2(1+D_ED_R)}\rightarrow0\,,
\nonumber
}
the distribution of non-zero eigenvalues $P(\lambda)$ of the states in ${\cal E}_{{\rm PPE}_R}$, generated from a Haar-random state, is asymptotically degenerate,
\eq{\lim_{N\rightarrow\infty}P(\lambda)\sim \delta(\lambda-\mu)\,,\label{eq:limit-plam}}
where $\mu=\max\{D_E^{-1},D_R^{-1}\}$. 

Given the asymptotic form of $P(\lambda)$ in Eq.~\ref{eq:limit-plam}, we now compute the Holevo information of ${\cal E}_{{\rm PPE}_R}$. We will consider two situations separately, $p<1-2\gamma$ and $p>1-2\gamma$, for reasons that will become clear shortly, starting with the former.  

\subsubsection{$p\leq 1-2\gamma$}
For $p\leq1-2\gamma$, the distribution of the set of eigenvalues of $\rho_R(o_S)$ is highly concentrated about $(1,1...1)/D_R$ for $N\gg 1$, which implies that the distribution of $\rho_R(o_S)$ is highly concentrated about $\mathbb{I}_{D_R}/D_R$. As $N\rightarrow\infty$, $\rho_R\rightarrow\mathbb{I}_{D_R}/D_R$ and its von Neumann entropy ${\cal S}_{\rm{vN}}(\rho_R(o_S))\rightarrow \ln D_R$, which yields 
\eq{
\chi({\cal E}_{{\rm PPE}_R})\rightarrow \ln D_R~(1-\sum_{o_S}p(o_S))=0\,.} 
This is precisely the measurement-invisible quantum-correlated phase discovered in Ref.~\cite{sherry2025miqc}.
The key point is that the ensemble ${\cal E}_{{\rm PPE}_R}$ has no non-trivial spread over the space of states.
This is manifested in the vanishing of the Holevo information.
While $\chi({\cal E}_{{\rm PPE}_R})$ vanishes identically in the limit of $N\to\infty$, it is also important to understand how it does so as the limit is approached. 
In the following we show that for $N\gg 1$, $\chi(\cal {E}_{\rm{PPE}_R})$ decays exponentially in $N$ in this $p$-$\gamma$ parameter space.

To understand this, it is useful to look at the second R\'enyi entropy ${\cal S}_2(\rho)=-\ln \rm{Tr}\rho^2$. 
Since $P(\lambda)$ is approximately degenerate as $N\gg 1$, the ensemble-averaged von Neumann entropy is approximately equal to the ensemble-averaged second R\'enyi entropy in the PPE. The Holevo information is therefore approximately given by 
\eq{\chi({\cal E}_{{\rm PPE}_R})\approx{\cal S}_{\rm{vN}}(\rho_R)-\sum_{o_S}p(o_S){\cal S}_2(\rho_R(o_S))\,,\label{eq:hol-bnd}}
where ${\cal S}_{\rm{vN}}(\rho_R(o_S))$ has been replaced by ${\cal S}_2(\rho_R(o_S))$ in the second term of the RHS. We note that the second term is now the ensemble average of the logarithm of the purity.
Instead of evaluating this `quenched' average, we evaluate the corresponding `annealed' average
\eq{
{\cal S}_2^{\rm ann}\equiv \ln\left[\sum_{o_S}p(o_S)~\text{Tr}\left[\rho^2_R(o_S)\right]\right]\,,
\label{eq:anneal}
}
which is approximately equivalent to the former for large $N$.
This is justified by the fact that, for $N\gg 1$, $\text{Tr}\left[\rho^2_R(o_S)\right]$ is approximately constant since $P(\lambda)$ is approximately degenerate.
To evaluate the annealed average in Eq.~\ref{eq:anneal}, we note that, from Eq.~\ref{eq:dens-trace} and Eq.~\ref{eq:dens-val-2}, we have
\eq{{\sum_{o_S}p(o_S)~\text{Tr}\left[\rho^2_R(o_S)\right]}=\frac{D_E+D_R}{1+D_RD_E}+\eta\,,}
For $N\gg1$, we know that
\eq{\eta\leq\epsilon \ll\frac{D_E+D_R}{1+D_RD_E}\,.}
This, along with the fact that $D_R\ll D_E$ in the MIQC phase, implies that the leading and subleading orders in $N$ of the annealed average is given by

\eq{\ln\left[\sum_{o_S}p(o_S)~\text{Tr}\left[\rho^2_R(o_S)\right]\right]\approx\ln D_R^{-1}+\frac{D_R^2-1}{1+D_RD_E}\,.\label{eq:anneal-approx}}
The leading and subleading orders in $N$ of the first term in the RHS of Eq.~\ref{eq:hol-bnd} can be obtained from the Page formula~\cite{page1993average}. The typical value of ${\cal S}_{\rm{vN}}(\rho_R)$ in the case where $R$, $S$ and $E$ are extensive subsystems with $|R|<|S\cup E|$ for $N\gg1$ is given by
\eq{{\cal S}_{\rm{vN}}(\rho_R)\approx\ln D_R-\frac{D_R}{2D_ED_S}\,.\label{eq:page-bound}}
Using Eq.~\ref{eq:page-bound} and Eq.~\ref{eq:anneal-approx}, Eq.~\ref{eq:hol-bnd} can be rewritten for $N\gg 1$ as
\eq{\chi({\cal E}_{{\rm PPE}_R})\approx \frac{D_R}{D_E}\left(1-\frac{1}{2D_S}\right)\,.}
Thus, $\chi({\cal E}_{{\rm PPE}_R})$ approximately decays as $2^{(2\gamma+p-1)N}$ for $N\gg1$ in the MIQC phase.

\subsubsection{$p>1-2\gamma$}

For $p>1-2\gamma$, the distribution of each \textit{non-zero} eigenvalue becomes highly concentrated about $1/D_E$ in the limit of $N\gg 1$.
Therefore the distribution of the set of eigenvalues of $\rho_R(o_S)$ is highly concentrated about 
\[(\underbrace{1,1....1}_{D_E},\underbrace{0,0....0}_{D_R-D_E})/D_E\]
(up to some ordering). 
Crucially, even though the von Neumann entropy ${\cal S}_{\rm{vN}}(\rho_R(o_S))$ is concentrated about $\ln D_E$, the distribution of $\rho_R(o_S)$ is spread out in the Hilbert space. As $N\rightarrow\infty$, ${\cal S}_{\rm{vN}}(\rho_R(o_S))\rightarrow \ln D_E$, which yields 
\eq{\begin{split}
    \chi({\cal E}_{{\rm PPE}_R})&\rightarrow N\ln 2~(\gamma-(1-p-\gamma)\sum_{o_S}p(o_S))\\
    &=N\ln 2~(p+2\gamma-1)\,.
\end{split}} 
This is precisely the manifestation of the {\it measurement-visible quantum-correlated} phase~\cite{sherry2025miqc} in the Holevo information.

At this juncture, note that the qualitatively different behaviour of Holevo information with $N$ for $p\leq 1-2\gamma$ and $p>1-2\gamma$.
In the former, $\chi({\cal E}_{{\rm PPE}_R})$ decays exponentially in $N$ to zero in the thermodynamic limit, whereas in the latter, it grows with $N$ linearly suggesting a `volume-law' behaviour. 
This is highly suggestive of an `information phase-transition' in the PPE in driven by the sizes of the subsystems, $p$ and $\gamma$. 
We will present concrete evidence for the same in the next subsection (Sec.~\ref{sec:pd-haar}) but before that we make a digression to discuss a physically pertinent limit.
 
The discussion so far has focused on the case where subsystems $R$, $S$ and $E$ are extensive. In the case where $R$ and $E$ are finite but $S$ is extensive, the Holevo information can be computed using the existing body on results on deep thermalisation. In this case, the PE of states supported on $R\cup E$ converges to the Haar ensemble as $|S|\rightarrow\infty$, in which case the von Neumann entropy of $R$ averaged over the PE is given by the Page formula~\cite{page1993average}
\eq{\langle{\cal S}_{\rm{vN}}(\rho_R(o_S)\rangle_{\rho_R(o_S)}=H_{mn}-H_n-\frac{m-1}{2n}\,,\label{eq:page-value}}
where $m=\min\{D_R,D_E\}$, $n=\max\{D_R,D_E\}$ and $H_{mn},H_n$ are harmonic numbers. The Holevo information of ${\cal E}_{\rm PPE}$ is obtained by subtracting Eq.~\ref{eq:page-value} from the von Neumann entropy of the average state,
\eq{\chi({\cal E}_{{\rm PPE}_R})=\ln m-H_{mn}+H_n+\frac{m-1}{2n}\,.}
While this limit is physically relevant, it firmly remains in the measurement-visible phase. The transition between the measurement-visible and measurement-invisible phases is observed as the sizes of the (extensive) subsystems are tuned, which we will now discuss.

\subsection{Information phase diagram \label{sec:pd-haar}}

\begin{figure}

    \includegraphics[width=\linewidth]{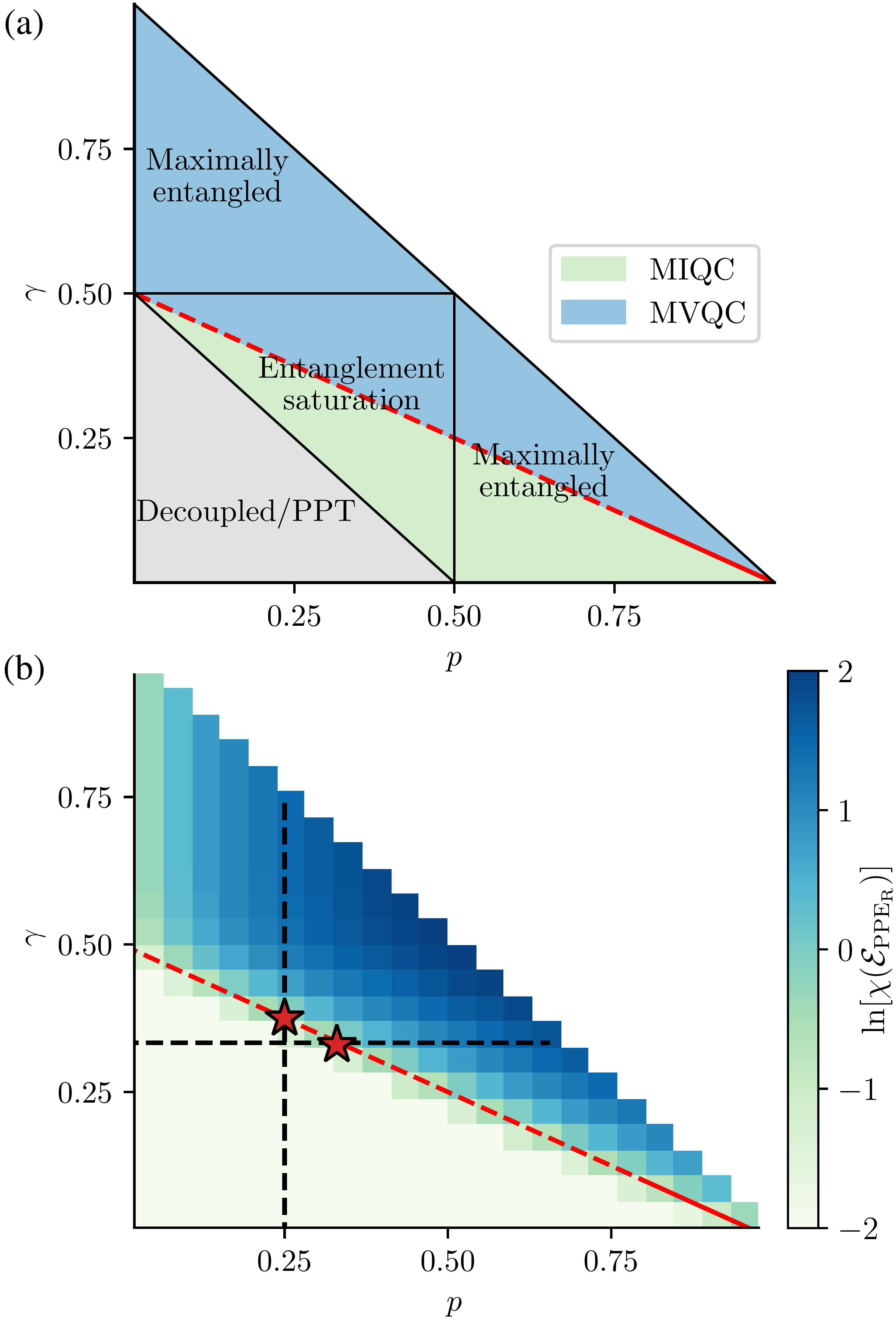}
    \caption{(a) Schematic phase diagram for the Holevo information $\chi({\cal E}_{{\rm PPE}_R})$ contrasted with the phase diagram of the entanglement between $R$ and $S$ (from \cite{shapourian2021negativity}). The entanglement phases are labelled, while the line of separation between the measurement-visible and measurement-invisible is denoted by the \as{red dashed-dotted} line, with the \as{solid} part representing what was obtained from the analytical calculations and the dashed-dotted part denoting the extrapolation using the numerical data. (b) Heatmap plot for $\chi({\cal E}_{{\rm PPE}_R})$, from which the data was used to extrapolate the line that separates the measurement-visible and measurement-invisible phases. The data was generated for tripartite Haar-random states over $24$ qubits at various partitions of $R$, $S$ and $E$ corresponding to various points $(p,\gamma)$. The black horizontal and vertical dashed lines indicate two representative slices for which finite-size scaling data is presented in Fig.~\ref{fig:Holevo_crossing}, with the red stars indicating the corresponding critical points.} 
    \label{fig:phase-dia}
\end{figure}

In the thermodynamic limit, as the sizes of the subsystems $R$ and $S$ (parametrised by $\gamma$ and $p$, respectively) are varied, the Holevo information undergoes a transition between two distinct phases. In one phase, $\chi({\cal E}_{{\rm PPE}_R})$ decays exponentially in $N$, while in the other, it follows a volume law (see Fig.~\ref{fig:phase-dia}(a)). 
In the regime $\max\{2\gamma,2-2p-2\gamma\}<1-8\gamma$, we have analytically demonstrated the existence of these two phases, referred to as the MIQC and MVQC phases, respectively. These phases are separated by the line $p = 1 - 2\gamma$, and numerical results indicate that this line can be extrapolated to intersect the $\gamma$-axis at $\gamma = 0.5$, extending beyond the analytically accessible region (see Fig.~\ref{fig:phase-dia}(b)).

\begin{figure}[!ht]
\includegraphics[width=\linewidth]{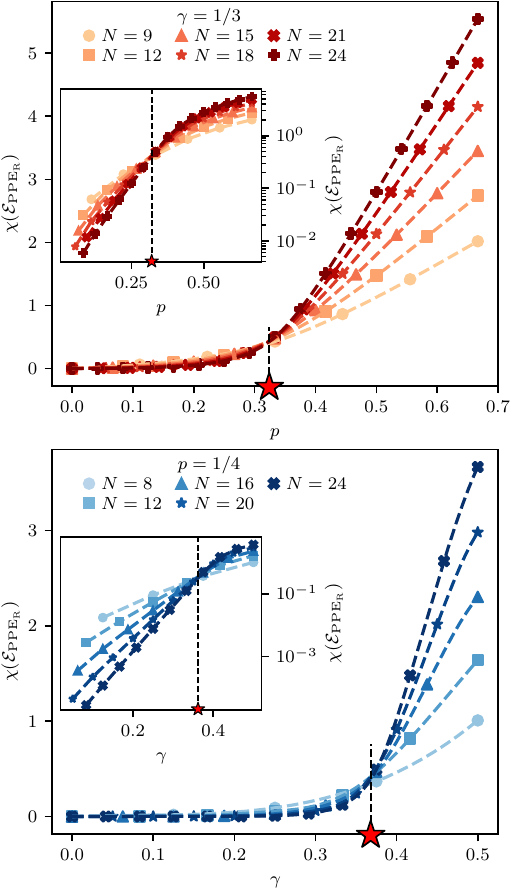}
\caption{Plots of $\chi(\mathcal{E}_{\mathrm{PPE_R}})$ across two different slices of the $p$-$\gamma$ parameter space for different system sizes $N$. The data shows a clear crossing for different $N$ indicating a transition. Top panel shows the data for fixed $\gamma=\frac{1}{3}$ with varying $p$. Bottom panel shows the transition for fixed $p=\frac{1}{4}$ with varying $\gamma$. These slices are marked in Fig.~\ref{fig:phase-dia}(b) by the black dashed horizontal and vertical lines respectively, and the critical points are marked by red stars as also indicated there. Insets show the same data but on logarithmic scales to highlight the exponential decay of $\chi(\mathcal{E}_{\mathrm{PPE_R}})$ with $N$ in the MIQC phase.}
\label{fig:Holevo_crossing}
\end{figure}

That the transition in the Holevo information is a genuine transition and not a crossover is demonstrated concretely by the finite-size scaling data for different slices  of the phase diagram. 
In Fig.~\ref{fig:Holevo_crossing} we show the data for 
 $\chi({\cal E}_{{\rm PPE}_R})$ for two representative slices (the slices are marked in Fig.~\ref{fig:phase-dia} (b) by the dashed lines), varying $p$ at fixed $\gamma$ and varying $\gamma$ at fixed $p$. 
 The data clearly shows a crossing of $\chi({\cal E}_{{\rm PPE}_R})$ for different $N$, which indicates a transition. In the MVQC phase, the data grows linearly with $N$ (as indicated by the equispaced data for different $N$ on a linear scale), whereas the data also shows the exponential decay of $\chi({\cal E}_{{\rm PPE}_R})$ in $N$ in the MIQC phase. This is evinced by the equispaced nature of the data on logarithmic scales (see insets). 

 These numerical results conclusively show that the transition between the MIQC and MVQC phases as diagnosed through the lens of the Holevo information manifests itself throughout the $p$-$\gamma$ parameter space and beyond the regime where we could prove the existence of the transition analytically in Sec.~\ref{sec:comp-chi}.
 The Holevo information transition line $p_c=1-2\gamma_c$ cuts through the negativity phase diagram of Haar-random states~\cite{shapourian2021negativity} and therefore presents a finer characterisation of the information structure of the state.
In fact, results presented in Appendix~\ref{app:ghse} suggest that the PPE continues to be well described by a gHSe, up to the second moment at least, along the Holevo information transition line and in the MIQC phase. 
On the other hand, in the MVQC phase our numerical data suggests that while the gHSe emerges in the regime of large $p$ and small $\gamma$ (lower right corner of the phase diagram) it may not do so at smaller $p$. This is concomitant with the fact that the PPE is not described a gHSe for $\gamma>0.5$ in the thermodynamic limit.

We now present a detailed contrast between the phases of the Holevo information and the entanglement phases of tripartite Haar-random states in \cite{shapourian2021negativity}. We briefly summarise the results of the latter here. The phases are described in the thermodynamic limit by the entanglement between subsystems $R$ and $S$, as quantified by the logarithmic negativity of ${\cal N}_{RS}$ in Eq.~\ref{eq:log-neg}. It has the following behaviour in three regimes of $p$-$\gamma$:
\begin{enumerate}
\item[(1)] If $p+\gamma<0.5$, the entanglement between $R$ and $S$ vanishes and the corresponding phase is described as the \textit{positive partial transpose}(PPT) or \textit{decoupled} phase. This follows from Page's formula~\cite{page1993average}, which implies that up to half the system is in the maximally mixed state. 

\item[(2)] If $\max\{p,\gamma\}>0.5$, the logarithmic negativity of $R$ and $S$ is given by $\min\{p,\gamma\}\times N$, which is the maximum possible entanglement and hence the corresponding phase is described as the \textit{maximally entangled} phase. The monogamy of entanglement implies that the smaller of subsystems $R$ and $S$ is decoupled from $E$. This also follows from Page's formula. 

\item[(3)] If each of the three subsystems $R$, $S$ and $E$ is smaller than half the system in size, the logarithmic negativity of subsystems $R$ and $S$ is given by~\cite{shapourian2021negativity} 
\eq{{\cal N}_{RS}=\frac{N}{2}(2p+2\gamma-1)+\ln\frac{8}{3\pi}\,.} This is the \textit{entanglement saturation} phase where subsystems $R$, $S$ and $E$ are genuinely tripartite entangled, and the entanglement between $R$ and $S$ is sub-maximal due to the monogamy of entanglement.
\end{enumerate}

Having summarised the phases of entanglement between $R$ and $S$, we now discuss the behaviour of the Holevo information $\chi({\cal E}_{{\rm PPE}_R})$ within each entanglement phase. In the PPT/decoupled phase, $\chi({\cal E}_{{\rm PPE}_R})$ obviously vanishes since the joint state of $R\cup S$ is a product state.

The maximally entangled phase with $\gamma > 0.5$, where ${\cal N}_{RS}$ is set by $|S|$, lies entirely in the MVQC phase. This is simply because in this regime, the entire entanglement of $S$ is shared with $R$, and $S$ is completely decoupled from $E$. Thus, the effect of the measurements on $S$ on $R$ cannot be mediated by $E$, leading to finite Holevo information. Indeed, the decoupling of the measured subsystem $S$ and the traced-out subsystem $E$ is a sufficient condition for measurement-visibility. 
On the other hand, the maximally entangled phase with ${\cal N}_{RS}\sim |R|$ as well as the entanglement saturation phase is split into two phases by the MIQC-MVQC transition line. This is because the MIQC phase arises in settings with highly scrambled entanglement, where the effect of the measurement of one subsystem ($S$), on another subsystem ($R$), is mediated by a third subsystem ($E$).

It is of interest to note that the MIQC phase has no counterpart in the context of bipartite pure states, since PEs of states on a subsystem generated by measurements on its entangled complement have non-zero Holevo information~\cite{sherry2025miqc}. 

Given the rigorously established phases of entanglement and the Holevo information (the latter in the region $\max\{2\gamma,2-2p-2\gamma\}<1-8\gamma$ of the phase diagram) in tripartite Haar-random states, the results of this work constitute a rigorous proof of existence of the MIQC and MVQC phases.

\section{Holevo information dynamics in chaotic circuits \label{sec:circuits}}
In this section, we demonstrate that the information phases and transitions between them, as encoded in the Holevo information of the PPE, emerges in scrambling dynamics effected by chaotic quantum circuits. 
One pertinent question in this context is regarding the timescales at which these phases appear and how do they scale with system size.

Our setting is the following. We consider 2-local random unitary circuits (with two different geometries as we discuss shortly) where the unitary time-evolution is built up of unitary gates acting on two qubits. 
Specifically, a gate, say acting between qubits $j$ and $k$ has the form 
\eq{
W_{jk}=
\cbox{
\begin{tikzpicture}
\foreach \x in {0,1}{
    \draw[thick] (\x,-0.35) -- (\x,0.35);
}
\def\r{.15}
\def\xa{0}
\def\ya{0}
\def\xb{1}
\def\yb{0}
\fill[Blue!80](\xa,\ya) circle (\r);
\fill[Blue!80](\xb,\yb) circle (\r);
\fill[Blue!80](\xa+0.866*\r,\ya-0.5*\r) rectangle (\xb-0.866*\r,\yb+0.5*\r);
\draw[thick] (\xa,\ya) ++(30:\r) arc[start angle=30, end angle=330, radius=\r];
\draw[thick] (\xb,\yb) ++(-150:\r) arc[start angle=-150, end angle=150, radius=\r];
\draw[thick](\xa+0.866*\r,\ya+0.5*\r) -- (\xb-0.866*\r,\yb+0.5*\r);
\draw[thick](\xa+0.866*\r,\ya-0.5*\r) -- (\xb-0.866*\r,\yb-0.5*\r);
\end{tikzpicture}}
= [u_j\otimes u_k]\exp(-i \tau H)[v_j\otimes v_k]\,,\label{eq:gate}
}
where $u_j,u_k,v_j,v_k$ are single-qubit Haar-random gates and 
\eq{
\begin{split}
H = 0.3X_jX_k + 0.2(X_j &+ X_k) + \\&0.4Z_jZ_k + 0.5(Z_j + Z_k)\,,
\end{split}
\label{eq:ham}
}
with $X_j (Z_j)$ denoting the Pauli-$X (Z)$ operator on site $j$.
The numerical parameters in Eq.~\ref{eq:ham} are not fine-tuned and the results do not depend qualitatively on these values.
The initial state is taken to be of the form
\eq{
\ket{\Psi(t=0)} = \otimes_{j=1}^N \ket{\varphi_j}\,,
}
where $\ket{\varphi_j}$ is a Haar-random state on the qubit $j$.
This initial state is evolved in time via a unitary circuit represented by the time-evolution operator ${\cal U}_t$.
For every $t$, we construct the PPE on $R$ from $\ket{\Psi(t)} = {\cal U}_t\ket{\Psi(0)}$ in exactly the same way as in Eq.~\ref{eq:ppe-def}, except now ${\cal E}_{{\rm PPE}_R}(t)$ is now explicitly time dependent.
As a function of $t$, we then track the Holevo information of the PPE, $\chi(t)\equiv \chi[{\cal E}_{{\rm PPE}_R}(t)]$.

Note that at $t=0$, since the state is a product state over all qubits, the PPE is trivial as $\rho_R(o_S,t=0) = \otimes_{j\in R}\ket{\varphi_j}\bra{\varphi_j}~\forall o_S$.
This trivially leads to $\chi(t=0)=0$.  
The signatures of the different information phases in the Holevo information emerges only at later times as the information in the global state is scrambled. 
In fact, as we shall shortly show, $\chi_{\rm sat}(N)\equiv  \lim_{t\to\infty}\chi(t,N)$ scales with $N$ as
\eq{
\chi_{\rm sat}(N) \sim
\begin{cases}
e^{-c N}\,; & {\rm MIQC}\\
N^0\,; & {\rm critical}\\
N\,; & {\rm MVQC}
\end{cases}\,,
\label{eq:chi-sat-scaling}
}
mirroring the results for PPEs generated from Haar-random states as discussed in Sec.~\ref{sec:phase-dia}.
The pertinent question then is what are timescales in the dynamics at which the scaling behaviour in Eq.~\ref{eq:chi-sat-scaling} set in in different phases and how do they depend on the geometry of the circuit. 
To address this we attempt to collapse the data of $\chi(t,N)$ for different $N$ onto a common curve of the form
\eq{
\Delta_\chi(t,N)\equiv\left|1-\frac{\chi(t,N)}{\chi_{\rm sat}(N)}\right| = g\left(\frac{t}{t_\ast(N)}\right)\,,
\label{eq:approach}
}
and study the scaling of $t_\ast(N)$ with $N$ as an estimate of the timescales at which the information phases emerge.
Note that the $t_\ast$ so estimated effectively tracks the timescale at which data for $\chi(t,N)$ approaches its $t\to\infty$ saturation value and hence can be understood as an upper bound for the timescales required to see the information phases and the transition. 
We now show results for $\chi(t,N)$ in two circuit geometries, an all-to-all circuit, and a $1+1$D brickwork circuit starting with the former. 

\subsection{All-to-all circuit}
We study the all-to-all circuit as the minimal chaotic model with no spatial structure which nevertheless encodes a notion of locality as all the gates are $2$-local. 
Formally, the time-evolution unitary can be written as 
\eq{
{\cal U}_t = \prod_{\vartheta=1}^t U_{\vartheta}\,,
}
where the time-evolution operator over one time step,
\eq{
U_{\vartheta} = \prod_{n=1}^N W_{i_nj_n}\,,
}
consists of $N$ gates of the form in Eq.~\ref{eq:gate} applied between randomly chosen pairs of qubits $(i_n,j_n)$.
Graphically, the circuit can be represented as,
\begin{tikzpicture}
\foreach \x in {0,...,7}{
    \draw[thick] (\x,-0.5) -- (\x,2.5);
    \draw[thick, dashed] (\x,2.5) -- (\x,3);
}
\def\r{.15}

\def\xa{0}
\def\ya{0}
\def\xb{3}
\def\yb{0}
\fill[Blue!80](\xa,\ya) circle (\r);
\fill[Blue!80](\xb,\yb) circle (\r);
\fill[Blue!80](\xa+0.866*\r,\ya-0.5*\r) rectangle (\xb-0.866*\r,\yb+0.5*\r);
\draw[thick] (\xa,\ya) ++(30:\r) arc[start angle=30, end angle=330, radius=\r];
\draw[thick] (\xb,\yb) ++(-150:\r) arc[start angle=-150, end angle=150, radius=\r];
\draw[thick](\xa+0.866*\r,\ya+0.5*\r) -- (\xb-0.866*\r,\yb+0.5*\r);
\draw[thick](\xa+0.866*\r,\ya-0.5*\r) -- (\xb-0.866*\r,\yb-0.5*\r);

\def\xa{4}
\def\ya{0}
\def\xb{6}
\def\yb{0}
\fill[YellowOrange!80](\xa,\ya) circle (\r);
\fill[YellowOrange!80](\xb,\yb) circle (\r);
\fill[YellowOrange!80](\xa+0.866*\r,\ya-0.5*\r) rectangle (\xb-0.866*\r,\yb+0.5*\r);
\draw[thick] (\xa,\ya) ++(30:\r) arc[start angle=30, end angle=330, radius=\r];
\draw[thick] (\xb,\yb) ++(-150:\r) arc[start angle=-150, end angle=150, radius=\r];
\draw[thick](\xa+0.866*\r,\ya+0.5*\r) -- (\xb-0.866*\r,\yb+0.5*\r);
\draw[thick](\xa+0.866*\r,\ya-0.5*\r) -- (\xb-0.866*\r,\yb-0.5*\r);

\def\xa{1}
\def\ya{0.3}
\def\xb{2}
\def\yb{0.3}
\fill[BrickRed!80](\xa,\ya) circle (\r);
\fill[BrickRed!80](\xb,\yb) circle (\r);
\fill[BrickRed!80](\xa+0.866*\r,\ya-0.5*\r) rectangle (\xb-0.866*\r,\yb+0.5*\r);
\draw[thick] (\xa,\ya) ++(30:\r) arc[start angle=30, end angle=330, radius=\r];
\draw[thick] (\xb,\yb) ++(-150:\r) arc[start angle=-150, end angle=150, radius=\r];
\draw[thick](\xa+0.866*\r,\ya+0.5*\r) -- (\xb-0.866*\r,\yb+0.5*\r);
\draw[thick](\xa+0.866*\r,\ya-0.5*\r) -- (\xb-0.866*\r,\yb-0.5*\r);

\def\xa{3}
\def\ya{0.3}
\def\xb{5}
\def\yb{0.3}
\fill[LimeGreen!80](\xa,\ya) circle (\r);
\fill[LimeGreen!80](\xb,\yb) circle (\r);
\fill[LimeGreen!80](\xa+0.866*\r,\ya-0.5*\r) rectangle (\xb-0.866*\r,\yb+0.5*\r);
\draw[thick] (\xa,\ya) ++(30:\r) arc[start angle=30, end angle=330, radius=\r];
\draw[thick] (\xb,\yb) ++(-150:\r) arc[start angle=-150, end angle=150, radius=\r];
\draw[thick](\xa+0.866*\r,\ya+0.5*\r) -- (\xb-0.866*\r,\yb+0.5*\r);
\draw[thick](\xa+0.866*\r,\ya-0.5*\r) -- (\xb-0.866*\r,\yb-0.5*\r);

\def\xa{6}
\def\ya{0.3}
\def\xb{7}
\def\yb{0.3}
\fill[Plum!80](\xa,\ya) circle (\r);
\fill[Plum!80](\xb,\yb) circle (\r);
\fill[Plum!80](\xa+0.866*\r,\ya-0.5*\r) rectangle (\xb-0.866*\r,\yb+0.5*\r);
\draw[thick] (\xa,\ya) ++(30:\r) arc[start angle=30, end angle=330, radius=\r];
\draw[thick] (\xb,\yb) ++(-150:\r) arc[start angle=-150, end angle=150, radius=\r];
\draw[thick](\xa+0.866*\r,\ya+0.5*\r) -- (\xb-0.866*\r,\yb+0.5*\r);
\draw[thick](\xa+0.866*\r,\ya-0.5*\r) -- (\xb-0.866*\r,\yb-0.5*\r);

\def\xa{0}
\def\ya{.6}
\def\xb{1}
\def\yb{0.6}
\fill[Melon](\xa,\ya) circle (\r);
\fill[Melon](\xb,\yb) circle (\r);
\fill[Melon](\xa+0.866*\r,\ya-0.5*\r) rectangle (\xb-0.866*\r,\yb+0.5*\r);
\draw[thick] (\xa,\ya) ++(30:\r) arc[start angle=30, end angle=330, radius=\r];
\draw[thick] (\xb,\yb) ++(-150:\r) arc[start angle=-150, end angle=150, radius=\r];
\draw[thick](\xa+0.866*\r,\ya+0.5*\r) -- (\xb-0.866*\r,\yb+0.5*\r);
\draw[thick](\xa+0.866*\r,\ya-0.5*\r) -- (\xb-0.866*\r,\yb-0.5*\r);

\def\xa{2}
\def\ya{.6}
\def\xb{4}
\def\yb{0.6}
\fill[NavyBlue!80](\xa,\ya) circle (\r);
\fill[NavyBlue!80](\xb,\yb) circle (\r);
\fill[NavyBlue!80](\xa+0.866*\r,\ya-0.5*\r) rectangle (\xb-0.866*\r,\yb+0.5*\r);
\draw[thick] (\xa,\ya) ++(30:\r) arc[start angle=30, end angle=330, radius=\r];
\draw[thick] (\xb,\yb) ++(-150:\r) arc[start angle=-150, end angle=150, radius=\r];
\draw[thick](\xa+0.866*\r,\ya+0.5*\r) -- (\xb-0.866*\r,\yb+0.5*\r);
\draw[thick](\xa+0.866*\r,\ya-0.5*\r) -- (\xb-0.866*\r,\yb-0.5*\r);

\def\xa{5}
\def\ya{.6}
\def\xb{7}
\def\yb{0.6}
\fill[Black!30](\xa,\ya) circle (\r);
\fill[Black!30](\xb,\yb) circle (\r);
\fill[Black!30](\xa+0.866*\r,\ya-0.5*\r) rectangle (\xb-0.866*\r,\yb+0.5*\r);
\draw[thick] (\xa,\ya) ++(30:\r) arc[start angle=30, end angle=330, radius=\r];
\draw[thick] (\xb,\yb) ++(-150:\r) arc[start angle=-150, end angle=150, radius=\r];
\draw[thick](\xa+0.866*\r,\ya+0.5*\r) -- (\xb-0.866*\r,\yb+0.5*\r);
\draw[thick](\xa+0.866*\r,\ya-0.5*\r) -- (\xb-0.866*\r,\yb-0.5*\r);

\def\xa{0}
\def\ya{1.0}
\def\xb{7}
\def\yb{1.0}
\fill[Orchid](\xa,\ya) circle (\r);
\fill[Orchid](\xb,\yb) circle (\r);
\fill[Orchid](\xa+0.866*\r,\ya-0.5*\r) rectangle (\xb-0.866*\r,\yb+0.5*\r);
\draw[thick] (\xa,\ya) ++(30:\r) arc[start angle=30, end angle=330, radius=\r];
\draw[thick] (\xb,\yb) ++(-150:\r) arc[start angle=-150, end angle=150, radius=\r];
\draw[thick](\xa+0.866*\r,\ya+0.5*\r) -- (\xb-0.866*\r,\yb+0.5*\r);
\draw[thick](\xa+0.866*\r,\ya-0.5*\r) -- (\xb-0.866*\r,\yb-0.5*\r);

\def\xa{1}
\def\ya{1.3}
\def\xb{4}
\def\yb{1.3}
\fill[Black!80](\xa,\ya) circle (\r);
\fill[Black!80](\xb,\yb) circle (\r);
\fill[Black!80](\xa+0.866*\r,\ya-0.5*\r) rectangle (\xb-0.866*\r,\yb+0.5*\r);
\draw[thick] (\xa,\ya) ++(30:\r) arc[start angle=30, end angle=330, radius=\r];
\draw[thick] (\xb,\yb) ++(-150:\r) arc[start angle=-150, end angle=150, radius=\r];
\draw[thick](\xa+0.866*\r,\ya+0.5*\r) -- (\xb-0.866*\r,\yb+0.5*\r);
\draw[thick](\xa+0.866*\r,\ya-0.5*\r) -- (\xb-0.866*\r,\yb-0.5*\r);

\def\xa{5}
\def\ya{1.3}
\def\xb{7}
\def\yb{1.3}
\fill[SpringGreen](\xa,\ya) circle (\r);
\fill[SpringGreen](\xb,\yb) circle (\r);
\fill[SpringGreen](\xa+0.866*\r,\ya-0.5*\r) rectangle (\xb-0.866*\r,\yb+0.5*\r);
\draw[thick] (\xa,\ya) ++(30:\r) arc[start angle=30, end angle=330, radius=\r];
\draw[thick] (\xb,\yb) ++(-150:\r) arc[start angle=-150, end angle=150, radius=\r];
\draw[thick](\xa+0.866*\r,\ya+0.5*\r) -- (\xb-0.866*\r,\yb+0.5*\r);
\draw[thick](\xa+0.866*\r,\ya-0.5*\r) -- (\xb-0.866*\r,\yb-0.5*\r);

\def\xa{0}
\def\ya{1.6}
\def\xb{5}
\def\yb{1.6}
\fill[Periwinkle](\xa,\ya) circle (\r);
\fill[Periwinkle](\xb,\yb) circle (\r);
\fill[Periwinkle](\xa+0.866*\r,\ya-0.5*\r) rectangle (\xb-0.866*\r,\yb+0.5*\r);
\draw[thick] (\xa,\ya) ++(30:\r) arc[start angle=30, end angle=330, radius=\r];
\draw[thick] (\xb,\yb) ++(-150:\r) arc[start angle=-150, end angle=150, radius=\r];
\draw[thick](\xa+0.866*\r,\ya+0.5*\r) -- (\xb-0.866*\r,\yb+0.5*\r);
\draw[thick](\xa+0.866*\r,\ya-0.5*\r) -- (\xb-0.866*\r,\yb-0.5*\r);

\def\xa{2}
\def\ya{1.9}
\def\xb{3}
\def\yb{1.9}
\fill[GreenYellow](\xa,\ya) circle (\r);
\fill[GreenYellow](\xb,\yb) circle (\r);
\fill[GreenYellow](\xa+0.866*\r,\ya-0.5*\r) rectangle (\xb-0.866*\r,\yb+0.5*\r);
\draw[thick] (\xa,\ya) ++(30:\r) arc[start angle=30, end angle=330, radius=\r];
\draw[thick] (\xb,\yb) ++(-150:\r) arc[start angle=-150, end angle=150, radius=\r];
\draw[thick](\xa+0.866*\r,\ya+0.5*\r) -- (\xb-0.866*\r,\yb+0.5*\r);
\draw[thick](\xa+0.866*\r,\ya-0.5*\r) -- (\xb-0.866*\r,\yb-0.5*\r);

\def\xa{4}
\def\ya{1.9}
\def\xb{7}
\def\yb{1.9}
\fill[Salmon](\xa,\ya) circle (\r);
\fill[Salmon](\xb,\yb) circle (\r);
\fill[Salmon](\xa+0.866*\r,\ya-0.5*\r) rectangle (\xb-0.866*\r,\yb+0.5*\r);
\draw[thick] (\xa,\ya) ++(30:\r) arc[start angle=30, end angle=330, radius=\r];
\draw[thick] (\xb,\yb) ++(-150:\r) arc[start angle=-150, end angle=150, radius=\r];
\draw[thick](\xa+0.866*\r,\ya+0.5*\r) -- (\xb-0.866*\r,\yb+0.5*\r);
\draw[thick](\xa+0.866*\r,\ya-0.5*\r) -- (\xb-0.866*\r,\yb-0.5*\r);

\def\xa{0}
\def\ya{2.2}
\def\xb{4}
\def\yb{2.2}
\fill[ForestGreen](\xa,\ya) circle (\r);
\fill[ForestGreen](\xb,\yb) circle (\r);
\fill[ForestGreen](\xa+0.866*\r,\ya-0.5*\r) rectangle (\xb-0.866*\r,\yb+0.5*\r);
\draw[thick] (\xa,\ya) ++(30:\r) arc[start angle=30, end angle=330, radius=\r];
\draw[thick] (\xb,\yb) ++(-150:\r) arc[start angle=-150, end angle=150, radius=\r];
\draw[thick](\xa+0.866*\r,\ya+0.5*\r) -- (\xb-0.866*\r,\yb+0.5*\r);
\draw[thick](\xa+0.866*\r,\ya-0.5*\r) -- (\xb-0.866*\r,\yb-0.5*\r);

\def\xa{5}
\def\ya{2.2}
\def\xb{6}
\def\yb{2.2}
\fill[Peach](\xa,\ya) circle (\r);
\fill[Peach](\xb,\yb) circle (\r);
\fill[Peach](\xa+0.866*\r,\ya-0.5*\r) rectangle (\xb-0.866*\r,\yb+0.5*\r);
\draw[thick] (\xa,\ya) ++(30:\r) arc[start angle=30, end angle=330, radius=\r];
\draw[thick] (\xb,\yb) ++(-150:\r) arc[start angle=-150, end angle=150, radius=\r];
\draw[thick](\xa+0.866*\r,\ya+0.5*\r) -- (\xb-0.866*\r,\yb+0.5*\r);
\draw[thick](\xa+0.866*\r,\ya-0.5*\r) -- (\xb-0.866*\r,\yb-0.5*\r);

\draw[thick,dotted](-0.2,0.8) -- (7.5,0.8);
\node at (8, 0.8) {$t=1$};
\draw[thick,dotted](-0.2,2.4) -- (7.5,2.4);
\node at (8, 2.4) {$t=2$};
\end{tikzpicture},
which also makes manifest the fact that the choice of geometry and subsystems is immaterial. 

\begin{figure}[!b]
\includegraphics[width=\linewidth]{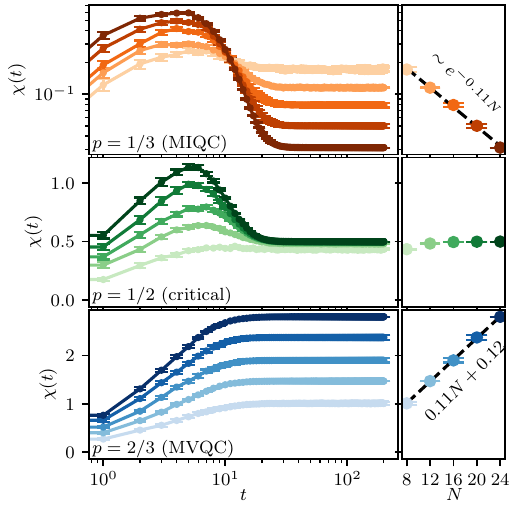}
\caption{Dynamics of the Holevo information, $\chi(t)$, for the 2-local all-to-all circuit. The data is for $\gamma=1/4$ and $\tau=0.5$. The three rows correspond to three different values of $p$, representative of the MIQC and MVQC phases and the critical point, as mentioned in the panels. The left column shows the dynamics of $\chi(t)$ as a function of $t$ for different $N=8,12,16,20,24$ (lighter to darker colours). The right columns show the infinite-time saturation value as a function of $N$. Note that the data for $p=1/3$ (MIQC phase) in the top row is on logarithmic scales. }
\label{fig:HI-alltoall}
\end{figure}

Representative results for the Holevo information in the two phases and the critical point is shown in Fig.~\ref{fig:HI-alltoall}; the left panels show $\chi(t)$ as a function of $t$ for different $N$, whereas the right panels show the saturation $\chi_{\rm sat}$ as a function of $N$.
In all three cases, $\chi(t)$ initially grows at early time.
This is due to the fact that in this regime, the quantum information is indeed shared between $R$ and $S$, but it is not sufficiently scrambled from the information phases to emerge.
In fact, we see that already at the earliest $O(1)$ times $\chi(t)$ grows with $N$.
We attribute this to absence of any spatial structure in the circuit.
As such within $O(1)$ times, with a finite probability, an extensive number of sites in $R$ are entangled with an extensive number of sites in $S$.
However, at later times, the information is significantly more scrambled such that tracing out $E$ has a qualitatively different effect on $\chi(t)$ depending on the phase.
This leads to the emergence of the MIQC and MVQC phases which is manifested clearly in the data saturating to values whose scaling with $N$ is in accordance with Eq.~\ref{eq:chi-sat-scaling}.

To understand the timescale $t_\ast(N)$ at which these different phases emerge, we study the approach of $\chi(t,N)$ to its saturation value $\chi_{\rm sat}(N)$, as in Eq.~\ref{eq:approach}. 
Within the limits of our numerical calculations we find that in the MIQC phase, $\Delta_\chi(t,N)$ for different $N$ collapses onto a common curve when plotted as a function of $t/\ln N$. On the other hand, at the critical point and in the MVQC phase, we find that the collapse is best when plotted as a function of just $t$. These collapses are shown in Fig.~\ref{fig:HI-timscale-alltoall}.
This suggests the scaling of $t_\ast$ with $N$ as 
\eq{
t_\ast(N)\sim \begin{cases}
\ln N\,; & {\rm MIQC}\\
N^0\,;& {\rm critical,MVQC}
\end{cases}\,.
\label{eq:tast-alltoall}
}
However, taking cognisance of the limited system sizes in our numerical calculations we cannot rule out that at much larger $N$, $t_\ast$ may become independent of $N$ in the MIQC phase as well.

\begin{figure}
\includegraphics[width=\linewidth]{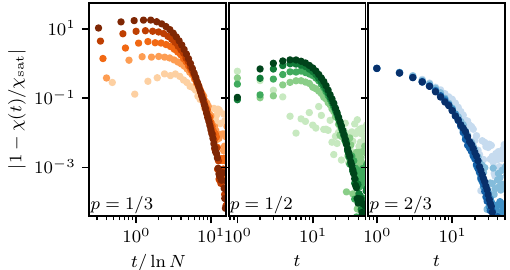}
\caption{Approach of $\chi(t)$ to its infinite-time value $\chi_{\rm sat}$ (see Eq.~\ref{eq:approach}) for the 2-local all-to-all circuit. The collapse of the data for different $N$ (different colour intensities, same as in Fig.~\ref{fig:HI-alltoall}) when plotted against $t/\ln N$ in the MIQC phase and against $t$ in the MVQC phase and the critical point suggests the scaling of $t_\ast$, the timescale associated to emergence of the information phases, to be of the form in Eq.~\ref{eq:tast-alltoall}.}
\label{fig:HI-timscale-alltoall}
\end{figure}

\subsection{1+1D brickwork circuit}

We next consider a 1+1D random circuit with brickwork geometry as a model which includes spatial locality as well in the simplest way. 
Since the gates act on nearest-neighbour qubits, there is a natural lightcone velocity which bounds how fast can information travel in this setting which in turn quatlitatively affects the scaling of $t_\ast$ with $N$ as we shall see shortly.
The time-evolution operator over one time step in this case,
\eq{
U_{\vartheta} = \prod_{n=1}^{N/2-1} W_{2n,2n+1}\prod_{n=1}^{N/2} W_{2n-1,2n}\,,
}
consists of a layer of gates acting only on the odd bonds followed by a layer of gates acting only on the even bonds.
This can be graphically depicted as 
\begin{tikzpicture}
\foreach \x in {0,...,7}{
    \draw[thick] (\x,-0.5) -- (\x,1.3);
    \draw[thick, dashed] (\x,1.3) -- (\x,1.8);
}
\draw [stealth-stealth, thick](-0.2,-0.7) -- (1.5,-0.7);
\draw [stealth-stealth, thick](1.5,-0.7) -- (5.5,-0.7);
\draw [stealth-stealth, thick](5.5,-0.7) -- (7.2,-0.7);
\node at (0.75, -1) {$R$};
\node at (3.5, -1) {$S$};
\node at (6.25, -1) {$E$};

\def\r{.15}

\def\xa{0}
\def\ya{0}
\def\xb{1}
\def\yb{0}
\fill[Blue!80](\xa,\ya) circle (\r);
\fill[Blue!80](\xb,\yb) circle (\r);
\fill[Blue!80](\xa+0.866*\r,\ya-0.5*\r) rectangle (\xb-0.866*\r,\yb+0.5*\r);
\draw[thick] (\xa,\ya) ++(30:\r) arc[start angle=30, end angle=330, radius=\r];
\draw[thick] (\xb,\yb) ++(-150:\r) arc[start angle=-150, end angle=150, radius=\r];
\draw[thick](\xa+0.866*\r,\ya+0.5*\r) -- (\xb-0.866*\r,\yb+0.5*\r);
\draw[thick](\xa+0.866*\r,\ya-0.5*\r) -- (\xb-0.866*\r,\yb-0.5*\r);

\def\xa{2}
\def\ya{0}
\def\xb{3}
\def\yb{0}
\fill[YellowOrange!80](\xa,\ya) circle (\r);
\fill[YellowOrange!80](\xb,\yb) circle (\r);
\fill[YellowOrange!80](\xa+0.866*\r,\ya-0.5*\r) rectangle (\xb-0.866*\r,\yb+0.5*\r);
\draw[thick] (\xa,\ya) ++(30:\r) arc[start angle=30, end angle=330, radius=\r];
\draw[thick] (\xb,\yb) ++(-150:\r) arc[start angle=-150, end angle=150, radius=\r];
\draw[thick](\xa+0.866*\r,\ya+0.5*\r) -- (\xb-0.866*\r,\yb+0.5*\r);
\draw[thick](\xa+0.866*\r,\ya-0.5*\r) -- (\xb-0.866*\r,\yb-0.5*\r);

\def\xa{4}
\def\ya{0}
\def\xb{5}
\def\yb{0}
\fill[BrickRed!80](\xa,\ya) circle (\r);
\fill[BrickRed!80](\xb,\yb) circle (\r);
\fill[BrickRed!80](\xa+0.866*\r,\ya-0.5*\r) rectangle (\xb-0.866*\r,\yb+0.5*\r);
\draw[thick] (\xa,\ya) ++(30:\r) arc[start angle=30, end angle=330, radius=\r];
\draw[thick] (\xb,\yb) ++(-150:\r) arc[start angle=-150, end angle=150, radius=\r];
\draw[thick](\xa+0.866*\r,\ya+0.5*\r) -- (\xb-0.866*\r,\yb+0.5*\r);
\draw[thick](\xa+0.866*\r,\ya-0.5*\r) -- (\xb-0.866*\r,\yb-0.5*\r);

\def\xa{6}
\def\ya{0}
\def\xb{7}
\def\yb{0}
\fill[LimeGreen!80](\xa,\ya) circle (\r);
\fill[LimeGreen!80](\xb,\yb) circle (\r);
\fill[LimeGreen!80](\xa+0.866*\r,\ya-0.5*\r) rectangle (\xb-0.866*\r,\yb+0.5*\r);
\draw[thick] (\xa,\ya) ++(30:\r) arc[start angle=30, end angle=330, radius=\r];
\draw[thick] (\xb,\yb) ++(-150:\r) arc[start angle=-150, end angle=150, radius=\r];
\draw[thick](\xa+0.866*\r,\ya+0.5*\r) -- (\xb-0.866*\r,\yb+0.5*\r);
\draw[thick](\xa+0.866*\r,\ya-0.5*\r) -- (\xb-0.866*\r,\yb-0.5*\r);

\def\xa{1}
\def\ya{0.3}
\def\xb{2}
\def\yb{0.3}
\fill[Plum!80](\xa,\ya) circle (\r);
\fill[Plum!80](\xb,\yb) circle (\r);
\fill[Plum!80](\xa+0.866*\r,\ya-0.5*\r) rectangle (\xb-0.866*\r,\yb+0.5*\r);
\draw[thick] (\xa,\ya) ++(30:\r) arc[start angle=30, end angle=330, radius=\r];
\draw[thick] (\xb,\yb) ++(-150:\r) arc[start angle=-150, end angle=150, radius=\r];
\draw[thick](\xa+0.866*\r,\ya+0.5*\r) -- (\xb-0.866*\r,\yb+0.5*\r);
\draw[thick](\xa+0.866*\r,\ya-0.5*\r) -- (\xb-0.866*\r,\yb-0.5*\r);

\def\xa{3}
\def\ya{.3}
\def\xb{4}
\def\yb{0.3}
\fill[Melon](\xa,\ya) circle (\r);
\fill[Melon](\xb,\yb) circle (\r);
\fill[Melon](\xa+0.866*\r,\ya-0.5*\r) rectangle (\xb-0.866*\r,\yb+0.5*\r);
\draw[thick] (\xa,\ya) ++(30:\r) arc[start angle=30, end angle=330, radius=\r];
\draw[thick] (\xb,\yb) ++(-150:\r) arc[start angle=-150, end angle=150, radius=\r];
\draw[thick](\xa+0.866*\r,\ya+0.5*\r) -- (\xb-0.866*\r,\yb+0.5*\r);
\draw[thick](\xa+0.866*\r,\ya-0.5*\r) -- (\xb-0.866*\r,\yb-0.5*\r);

\def\xa{5}
\def\ya{.3}
\def\xb{6}
\def\yb{0.3}
\fill[NavyBlue!80](\xa,\ya) circle (\r);
\fill[NavyBlue!80](\xb,\yb) circle (\r);
\fill[NavyBlue!80](\xa+0.866*\r,\ya-0.5*\r) rectangle (\xb-0.866*\r,\yb+0.5*\r);
\draw[thick] (\xa,\ya) ++(30:\r) arc[start angle=30, end angle=330, radius=\r];
\draw[thick] (\xb,\yb) ++(-150:\r) arc[start angle=-150, end angle=150, radius=\r];
\draw[thick](\xa+0.866*\r,\ya+0.5*\r) -- (\xb-0.866*\r,\yb+0.5*\r);
\draw[thick](\xa+0.866*\r,\ya-0.5*\r) -- (\xb-0.866*\r,\yb-0.5*\r);

\def\xa{0}
\def\ya{0.7}
\def\xb{1}
\def\yb{0.7}
\fill[Orchid](\xa,\ya) circle (\r);
\fill[Orchid](\xb,\yb) circle (\r);
\fill[Orchid](\xa+0.866*\r,\ya-0.5*\r) rectangle (\xb-0.866*\r,\yb+0.5*\r);
\draw[thick] (\xa,\ya) ++(30:\r) arc[start angle=30, end angle=330, radius=\r];
\draw[thick] (\xb,\yb) ++(-150:\r) arc[start angle=-150, end angle=150, radius=\r];
\draw[thick](\xa+0.866*\r,\ya+0.5*\r) -- (\xb-0.866*\r,\yb+0.5*\r);
\draw[thick](\xa+0.866*\r,\ya-0.5*\r) -- (\xb-0.866*\r,\yb-0.5*\r);

\def\xa{2}
\def\ya{0.7}
\def\xb{3}
\def\yb{0.7}
\fill[Black!80](\xa,\ya) circle (\r);
\fill[Black!80](\xb,\yb) circle (\r);
\fill[Black!80](\xa+0.866*\r,\ya-0.5*\r) rectangle (\xb-0.866*\r,\yb+0.5*\r);
\draw[thick] (\xa,\ya) ++(30:\r) arc[start angle=30, end angle=330, radius=\r];
\draw[thick] (\xb,\yb) ++(-150:\r) arc[start angle=-150, end angle=150, radius=\r];
\draw[thick](\xa+0.866*\r,\ya+0.5*\r) -- (\xb-0.866*\r,\yb+0.5*\r);
\draw[thick](\xa+0.866*\r,\ya-0.5*\r) -- (\xb-0.866*\r,\yb-0.5*\r);

\def\xa{4}
\def\ya{.7}
\def\xb{5}
\def\yb{.7}
\fill[SpringGreen](\xa,\ya) circle (\r);
\fill[SpringGreen](\xb,\yb) circle (\r);
\fill[SpringGreen](\xa+0.866*\r,\ya-0.5*\r) rectangle (\xb-0.866*\r,\yb+0.5*\r);
\draw[thick] (\xa,\ya) ++(30:\r) arc[start angle=30, end angle=330, radius=\r];
\draw[thick] (\xb,\yb) ++(-150:\r) arc[start angle=-150, end angle=150, radius=\r];
\draw[thick](\xa+0.866*\r,\ya+0.5*\r) -- (\xb-0.866*\r,\yb+0.5*\r);
\draw[thick](\xa+0.866*\r,\ya-0.5*\r) -- (\xb-0.866*\r,\yb-0.5*\r);

\def\xa{6}
\def\ya{.7}
\def\xb{7}
\def\yb{.7}
\fill[Periwinkle](\xa,\ya) circle (\r);
\fill[Periwinkle](\xb,\yb) circle (\r);
\fill[Periwinkle](\xa+0.866*\r,\ya-0.5*\r) rectangle (\xb-0.866*\r,\yb+0.5*\r);
\draw[thick] (\xa,\ya) ++(30:\r) arc[start angle=30, end angle=330, radius=\r];
\draw[thick] (\xb,\yb) ++(-150:\r) arc[start angle=-150, end angle=150, radius=\r];
\draw[thick](\xa+0.866*\r,\ya+0.5*\r) -- (\xb-0.866*\r,\yb+0.5*\r);
\draw[thick](\xa+0.866*\r,\ya-0.5*\r) -- (\xb-0.866*\r,\yb-0.5*\r);

\def\xa{1}
\def\ya{1}
\def\xb{2}
\def\yb{1}
\fill[Black!30](\xa,\ya) circle (\r);
\fill[Black!30](\xb,\yb) circle (\r);
\fill[Black!30](\xa+0.866*\r,\ya-0.5*\r) rectangle (\xb-0.866*\r,\yb+0.5*\r);
\draw[thick] (\xa,\ya) ++(30:\r) arc[start angle=30, end angle=330, radius=\r];
\draw[thick] (\xb,\yb) ++(-150:\r) arc[start angle=-150, end angle=150, radius=\r];
\draw[thick](\xa+0.866*\r,\ya+0.5*\r) -- (\xb-0.866*\r,\yb+0.5*\r);
\draw[thick](\xa+0.866*\r,\ya-0.5*\r) -- (\xb-0.866*\r,\yb-0.5*\r);

\def\xa{3}
\def\ya{1}
\def\xb{4}
\def\yb{1}
\fill[GreenYellow](\xa,\ya) circle (\r);
\fill[GreenYellow](\xb,\yb) circle (\r);
\fill[GreenYellow](\xa+0.866*\r,\ya-0.5*\r) rectangle (\xb-0.866*\r,\yb+0.5*\r);
\draw[thick] (\xa,\ya) ++(30:\r) arc[start angle=30, end angle=330, radius=\r];
\draw[thick] (\xb,\yb) ++(-150:\r) arc[start angle=-150, end angle=150, radius=\r];
\draw[thick](\xa+0.866*\r,\ya+0.5*\r) -- (\xb-0.866*\r,\yb+0.5*\r);
\draw[thick](\xa+0.866*\r,\ya-0.5*\r) -- (\xb-0.866*\r,\yb-0.5*\r);

\def\xa{5}
\def\ya{1}
\def\xb{6}
\def\yb{1}
\fill[Salmon](\xa,\ya) circle (\r);
\fill[Salmon](\xb,\yb) circle (\r);
\fill[Salmon](\xa+0.866*\r,\ya-0.5*\r) rectangle (\xb-0.866*\r,\yb+0.5*\r);
\draw[thick] (\xa,\ya) ++(30:\r) arc[start angle=30, end angle=330, radius=\r];
\draw[thick] (\xb,\yb) ++(-150:\r) arc[start angle=-150, end angle=150, radius=\r];
\draw[thick](\xa+0.866*\r,\ya+0.5*\r) -- (\xb-0.866*\r,\yb+0.5*\r);
\draw[thick](\xa+0.866*\r,\ya-0.5*\r) -- (\xb-0.866*\r,\yb-0.5*\r);

\draw[thick,dotted](-0.2,0.5) -- (7.5,0.5);
\node at (8, 0.5) {$t=1$};
\draw[thick,dotted](-0.2,1.2) -- (7.5,1.2);
\node at (8, 1.2) {$t=2$};
\end{tikzpicture}\,.
Given the spatial structure of the circuit, it is natural that the non-universal early-time behaviour is sensitive to the geometry and the spatial locations of the subsystems $R$, $S$, and $E$.
In accordance with the motivation of studying the effect of spatial locality we consider the subsystems to be contiguous as shown above.

\begin{figure}[!b]
\includegraphics[width=\linewidth]{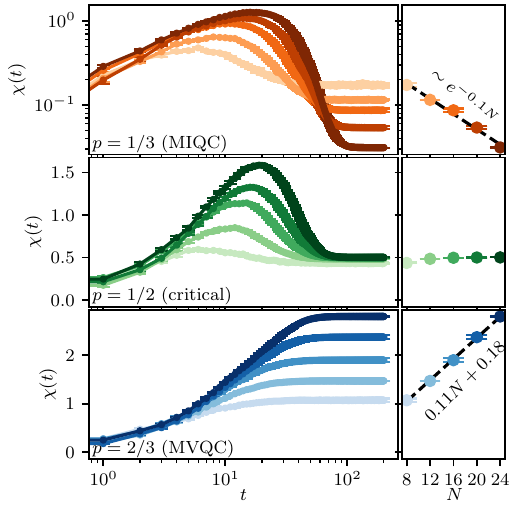}
\caption{Dynamics of the Holevo information, $\chi(t)$, for the 1+1D brickwork circuit. The data is for $\gamma=1/4$ and $\tau=0.8$. The three rows correspond to three different values of $p$ (same as in Fig.~\ref{fig:HI-alltoall}), representative of the MIQC and MVQC phases and the critical point. The left column shows the dynamics of $\chi(t)$ as a function of $t$ for different $N=8,12,16,20,24$ (lighter to darker colours), whereas the the right column shows the infinite-time saturation value as a function of $N$. The exponential decay and linear growth with $N$ of $\chi_{\rm sat}(N)$ is evident from the fits denoted by the black dashed lines.}
\label{fig:HI-bw}
\end{figure}

The results for the dynamics of $\chi(t)$ are shown in Fig.~\ref{fig:HI-bw}.
The universal features of the dynamics are qualitatively very similar to those of the all-to-all circuit as in Fig.~\ref{fig:HI-alltoall}.
At late times, in the MIQC and MVQC phases, $\chi_{\rm sat}$ decays exponentially and grows linearly, respectively, with $N$, whereas at the critical point, it is independent of $N$. 
However, at early times, the behaviour of $\chi(t)$ is distinct from that of the all-to-all circuit, in that the data is converged with $N$.
This is a reflection of the fact that the spatial locality of the circuit bounds the speed at which information is scrambled and spreads in space. 

The effect of this on the timescales for the information phases to emerge is that $t_\ast(N)\sim N$ in both the phases as well as the critical point. 
This is demonstrated in Fig.~\ref{fig:HI-timscale-bw} where $\Delta_{\chi}(t,N)$ (see Eq.~\ref{eq:approach}) when plotted as a function of $t/N$ for different $N$ collapse onto each other. 
Note that this is entirely consistent with the MIQC-MVQC transition happening at finite {\it depth densities} in 1+1D circuits, albeit in a slightly different setting~\cite{sherry2025miqc}.

\begin{figure}
\includegraphics[width=\linewidth]{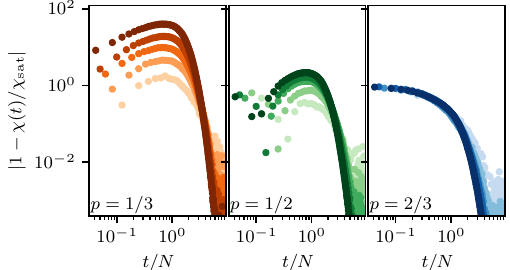}
\caption{For the 1+1D brickwork, the data for $\chi(t,N)$ approaches its infinite-time value $\chi_{\rm sat}(N)$ following a function of $t/N$, implying $t_\ast\sim N$, in both the phases as well as the critical point. This is indicated by $\Delta_\chi(t,N)$,  defined in Eq.~\ref{eq:approach}, for different $N$ (different colour intensities) collapsing onto each other when plotted as a function of $t/N$.}
\label{fig:HI-timscale-bw}
\end{figure}

\section{Summary and Outlook \label{sec:summary}}
We now briefly summarise the main results of this work. We have analysed the information content of partial projected ensembles (PPEs) generated from Haar-random states in the thermodynamic limit, quantified via the Holevo information. The PPE is defined as an ensemble of states over a subsystem $R$, conditioned upon measurements on another subsystem $S$, both of which are coupled to a bath/traced-out subsystem $E$. By leveraging the emergence of the second moment of the generalised Hilbert–Schmidt ensemble (gHSe) in the PPE, we analytically show that the spectral density of the states in the PPE are degenerate in the large-$N$ limit. Since the Holevo information depends only on these spectra, this allows us to compute it exactly in the large-$N$ limit. Together with supporting numerical results, this analysis reveals a sharp transition in the scaling of the Holevo information as the subsystem sizes are varied: when $R$ is smaller than the bath $E$, the Holevo information decays exponentially with system size, defining a \emph{measurement-invisible} phase. Once the size of $R$ exceeds that of $E$, the scaling of the Holevo information transitions into a volume-law regime, which is the \emph{measurement-visible} phase.

We contrasted the behaviour of the Holevo information with the entanglement phases of tripartite Haar-random states and found that the information structure of the PPE differs markedly from that of the underlying state. Notably, the Holevo phase boundary cuts across entanglement phase boundaries, revealing a \emph{measurement-invisible quantum-correlated} (MIQC) phase, in which extensive entanglement coexists with vanishing measurement influence between subsystems. This phenomenon, absent in bipartite pure states, is emblematic of many-body entanglement scrambling. Using our analytical approach, we rigorously established the existence of the MIQC phase for tripartite Haar-random states. The phase diagram further indicates that a necessary condition for the emergence of the MIQC phase is that the measured subsystem $S$ must be entangled with the bath $E$, highlighting the role of a mediating third subsystem in screening measurement back-action.
\add{A heuristic but intuitive picture that therefore emerges is that the decoupling regime arises due to the fact that all of the entanglement between $R$ and $S$ is mediated by $E$ such that tracing out $E$ decouples $R$ and $S$ (or equivalently,
$R$ and $S$ are decoupled due to monogamy of entanglement since they are individually maximally entangled with $E$). However, a picture that relies solely on monogamy is insufficient to explain the emergence of the MIQC phase. This is because in this phase $E$ mediates all of the effects of measurement of $S$ on $R$, but not all of the entanglement, as evidenced by the non-vanishing negativity in the MIQC
phase.}

Having established the phases of the Holevo information and the transitions between them for Haar-random states, we numerically demonstrated that precisely the same phase diagram emerges dynamically in circuits built from 2-local gates, with no spatial structure (all-to-all circuit) and with a 1+1D brickwork spatial structure.
In the former the information phases emerged at timescales which had either no or very weak (logarithmic) dependence on system size, whereas in the latter, the timescales scaled linearly with system size.

Our results on the emergence of the MIQC phase have implications for deep thermalisation, in particular for projected ensembles over extensive subsystems in settings with qubit loss errors or lossy measurements.  While one anticipates that the (partial) projected ensemble on $R$ approaches the trivial ensemble ($\rho_R(o_S)=\rho_R~\forall~o_S$) only once the entanglement between $S$ and $R$ becomes small, our results show that the trivial ensemble can emerge even when this entanglement persists at extensive scales.

While it is evident from our results that the MIQC phase arises due to the scrambled nature of quantum correlations in highly entangled many-body states, it is additionally insightful to view the emergence of the MIQC phase from the perspective of typicality phenomena. The central idea is that reduced density matrices of states drawn from broadly distributed ensembles are typical, as exemplified by canonical typicality~\cite{goldstein2006canonical}. This phenomenon arises because an overwhelming majority of states in the Hilbert space (or within fixed-energy subspaces) are highly entangled volume-law states with nearly maximal-entropy reduced density matrices. In our setting, deep thermalisation at low-order moments suggests that projected ensembles sample a broad distribution, implying that PPEs become trivial in certain regimes determined by subsystem sizes. One may therefore expect the phenomenology of trivial ensembles to persist even in scenarios where PPEs do not converge to the gHSe, even at low-order moments. For instance, in systems with conservation laws, PPEs may converge to partial traces of Scrooge ensembles~\cite{mark2024maximum}, and a recent generalisation of canonical typicality~\cite{Teufel2024gencanonical} demonstrates that reduced density matrices of Scrooge-ensemble states are themselves typical. This suggests that the MIQC phase might persist in the presence of conservation laws, which would be an interesting direction for future work.

\begin{acknowledgements}
AS would like to thank W. W. Ho for insightful discussions. SM and SR would like to thank P. W. Claeys for useful discussions and collaboration on related work. The authors acknowledge support of the Department of Atomic Energy, Government of India, under project no. RTI4001.
SR acknowledges support from SERB-DST, Government of India, under Grant No. SRG/2023/000858 and from a Max Planck Partner Group grant between ICTS-TIFR, Bengaluru and MPIPKS, Dresden.
\end{acknowledgements}

\onecolumngrid
\appendix
\section{Proof of Theorem~\ref{thm:ghse} \label{app:ghse-thm}}
We stated in Theorem~\ref{thm:ghse} that, for a tripartite Haar-random state $\ket{\Psi}$, the moments of ${\cal E}_{{\rm PPE}_R}$ satisfy $\Vert \rho^{(k)}_{{\rm{PPE}}_R}-\rho^{(k)}_{{\rm{gHSe}}_R}\Vert_1\leq \epsilon$ with probability at least $1-\delta$ if the bound in Eq.~\ref{eq:ghse-bound} is satisfied. In this appendix, we provide a proof for this bound.
\begin{proof}
We first lay out a few definitions and state a couple of lemmas. The structure of the proof relies on similar theorems proved in \cite{cotler2023emergent}. We define:
	\begin{equation*}
\ket{\widetilde{\psi}_z}=\left(\mathbb{I}_R\otimes\ket{z}\bra{z}_S\otimes\mathbb{I}_E\right)\ket{\Psi}\,,
	\end{equation*}
	such that the $k^{th}$ moment of projected ensemble of mixed states on $R$ is given by 
	\begin{equation*}
		\rho^{(k)}_{{\rm{PPE}}_R}	= \sum_z\frac{\left(\text{Tr}_E\left(\ket{\widetilde{\psi}_z}\bra{\widetilde{\psi}_z}\right)\right)^{\otimes k}}{\braket{\widetilde{\psi}_z|\widetilde{\psi}_z}^{k-1}}\,.
		\end{equation*}
	\textbf{Lemma 1 (Levy's lemma).} Let $f:\mathbb{S}^{2d-1}\rightarrow \mathbb{R}$ satisfying $|f(u)-f(v)|\leq \eta~\Vert u-v\Vert_2$. Then, for any $\epsilon\geq 0$, we have
	\begin{equation}
		\text{Prob}_{\Phi\sim\text{Haar}(d)}\left[\vert f(\Phi)-\mathbb{E}_{\Psi\sim\text{Haar}(d)}\left[f(\Psi)\right]\vert\geq \epsilon\right]\leq 2\exp\left(-\frac{2d\epsilon^2}{9\pi^3\eta^2}\right)\,.
	\end{equation}
	We now define the function which enables us to apply Levy's lemma.  Let $\{\ket{i}\}$ be a basis of $\mathcal{H}_R^{\otimes k}$ formed by the tensor product of basis states of each copy of $\mathcal{H}_R$ such that $\ket{i}=\otimes_{l=1}^k\ket{i^{(l)}}$. Similarly, we write $\ket{j}=\otimes_{l=1}^k\ket{j^{(l)}}$. Analogously, let $\{\ket{r}\}$ be a basis of $\mathcal{H}_E^{\otimes k}$ formed by the tensor product of basis states of each copy of $\mathcal{H}_R$ such that $\ket{r}=\otimes_{l=1}^k\ket{r^{(l)}}$. This will help us take the partial trace. Now, consider the function $f_{ij}:\mathbb{S}^{2d-1}\rightarrow\mathbb{R}$ defined by
	\begin{equation}
		\begin{split}
		f_{ij}(\ket{\Psi})&=\bra{i}\left(\sum_z\frac{\left(\text{Tr}_E\left(\ket{\widetilde{\psi}_z}\bra{\widetilde{\psi}_z}\right)\right)^{\otimes k}}{\braket{\widetilde{\psi}_z|\widetilde{\psi}_z}^{k-1}}\right)\ket{j}\\
		&=\sum_z\frac{\prod_{l=1}^k\sum_{r^{(l)}}\braket{i^{(l)},r^{(l)}|\widetilde{\psi}_z}\braket{\widetilde{\psi}_z|j^{(l)},r^{(l)}}}{\braket{\widetilde{\psi}_z|\widetilde{\psi}_z}^{k-1}}\,.
		\end{split}
		\label{eq:fij}
	\end{equation}
	\\\\
	\textbf{Lemma 2.} A Lipschitz constant for $f_{ij}$ is $\eta=2(2k-1)$.
	\\\\
 Borrowing the notational framework from \cite{cotler2023emergent}, let $\ket{\Psi}=\begin{bmatrix}
     \mathbb{I} & i\mathbb{I}
 \end{bmatrix}.~\Vec{u}$, where $\Vec{u}\in\mathbb{S}^{2d-1}$ and $\mathbb{I}$ represents the $d\times d$ identity matrix.   
Any $\eta$ such that
\begin{equation}
	\eta\geq \left\Vert\frac{d}{d\Vec{u}}f_{ij}(\ket{\Psi})\right\Vert_2\,,
\end{equation}
can be a Lipschitz constant since $f_{ij}$ is differentiable. 
\as{The vector derivative of $f_{ij}$ is defined as 
\eq{\frac{d}{d\Vec{u}}f_{ij}(\ket{\Psi})=\left(\frac{\partial f_{ij}}{\partial u^1}(\ket{\Psi}),~...,~\frac{\partial f_{ij}}{\partial u^{2d}}(\ket{\Psi})\right)}
and the Euclidean norm of an operator $O$ is defined as $\Vert O\Vert_2={\rm{Tr}}\left[O^\dagger O\right]$.} We now have
\begin{equation}
\begin{split}
\left\Vert\frac{d}{d\Vec{u}}f_{ij}(\ket{\Psi})\right\Vert_2 =&~\left\Vert\sum_{l'=1}^k\sum_z\frac{\prod_{l\neq l'}^k\sum_{r^{(l)}}\braket{i^{(l)},r^{(l)}|\widetilde{\psi}_z}\braket{\widetilde{\psi}_z|j^{(l)},r^{(l)}}}{\braket{\widetilde{\psi}_z|\widetilde{\psi}_z}^{k-1}}\sum_{r^{(l')}}\frac{d}{d\Vec{u}}\braket{i^{(l)},r^{(l)}|\widetilde{\psi}_z}\braket{\widetilde{\psi}_z|j^{(l)},r^{(l)}}\right.\\
&\left.-~(k-1)\sum_z\frac{\prod_{l=1}^k\sum_{r^{(l)}}\braket{i^{(l)},r^{(l)}|\widetilde{\psi}_z}\braket{\widetilde{\psi}_z|j^{(l)},r^{(l)}}}{\braket{\widetilde{\psi}_z|\widetilde{\psi}_z}^{k}}\frac{d}{d\Vec{u}}\braket{\widetilde{\psi}_z|\widetilde{\psi}_z}\right\Vert_2\\
\leq &~\left\Vert\sum_{l'=1}^k\sum_z\frac{\prod_{l\neq l'}^k\sum_{r^{(l)}}\braket{i^{(l)},r^{(l)}|\widetilde{\psi}_z}\braket{\widetilde{\psi}_z|j^{(l)},r^{(l)}}}{\braket{\widetilde{\psi}_z|\widetilde{\psi}_z}^{k-1}}\left(\begin{bmatrix}
    M_{l'}^+ & -iM_{l'}^-\\
    iM_{l'}^- & M^+_{l'}
\end{bmatrix}\otimes\ket{z}\bra{z}_S\otimes\mathbb{I}_E\right).~\Vec{u}~\right\Vert_2\\
&+2(k-1)\left\Vert~\sum_z\frac{\prod_{l=1}^k\sum_{r^{(l)}}\braket{i^{(l)},r^{(l)}|\widetilde{\psi}_z}\braket{\widetilde{\psi}_z|j^{(l)},r^{(l)}}}{\braket{\widetilde{\psi}_z|\widetilde{\psi}_z}^{k}}\begin{bmatrix}
    P_z & 0 \\
    0 & P_z
\end{bmatrix}.~\Vec{u}~
\right\Vert_2\,,
\label{eq:ineq-lemma2}
\end{split}
\end{equation}
where we have defined $M^{\pm}_{l'}=\ket{i^{(l')}}\bra{j^{(l')}}\pm\ket{j^{(l')}}\bra{i^{(l')}}$. The latter line arises from the triangle inequality and explicit evaluation of the derivatives. Defining 
\begin{equation}
\begin{split}
b_{l',z}&=\frac{\prod_{l\neq l'}^k\sum_{r^{(l)}}\braket{i^{(l)},r^{(l)}|\widetilde{\psi}_z}\braket{\widetilde{\psi}_z|j^{(l)},r^{(l)}}}{\braket{\widetilde{\psi}_z|\widetilde{\psi}_z}^{k-1}}\\
c_z&=\frac{\prod_{l=1}^k\sum_{r^{(l)}}\braket{i^{(l)},r^{(l)}|\widetilde{\psi}_z}\braket{\widetilde{\psi}_z|j^{(l)},r^{(l)}}}{\braket{\widetilde{\psi}_z|\widetilde{\psi}_z}^{k}}\\
M_{l'}&=\begin{bmatrix}
    M_{l'}^+ & -iM_{l'}^-\\
    iM_{l'}^- & M^+_{l'}\\
    \end{bmatrix}\,,   
\end{split}
\end{equation}
Eq.~\eqref{eq:ineq-lemma2} can be rewritten as
\begin{equation}
    \begin{split}
        \left\Vert\frac{d}{d\Vec{u}}f_{ij}(\ket{\Psi})\right\Vert_2 \leq &~  \left(\sum_{l'=1,p'=1}^k\sum_zb^*_{p',z}b_{l',z}~\Vec{u}^T.\left(M^\dagger_{p'}M_{l'}\otimes\ket{z}\bra{z}_S\otimes\mathbb{I}_E\right).~\Vec{u}\right)^{1/2}\\
        &+2(k-1)\left(\sum_z|c_z|^2~\Vec{u}^T.\begin{bmatrix}
            P_z & 0\\
            0 & P_z        \end{bmatrix}.~\Vec{u}\right)^{1/2}\\
        \leq &~  \left(\sum_{l'=1,p'=1}^k\sum_z|b^*_{p',z}||b_{l',z}|~\Vec{u}^T.\left(|M^\dagger_{p'}M_{l'}|\otimes\ket{z}\bra{z}_S\otimes\mathbb{I}_E\right).~\Vec{u}\right)^{1/2}\\
        &+2(k-1)\left(\sum_z|c_z|^2~\Vec{u}^T.\begin{bmatrix}
            P_z & 0\\
            0 & P_z        \end{bmatrix}.~\Vec{u}\right)^{1/2}\,, 
    \label{eq:lemma2-ineq2}
    \end{split}
\end{equation}
where $|A|=\sqrt{A^\dagger A}$. Now since
\begin{equation}
\begin{split}
|b_{l',z}|&=\left\vert\text{Tr}\left\{\left(\otimes_{l\neq l'}\ket{j^{(l)}}\bra{i^{(l)}}\right).\left(\frac{\left(\text{Tr}_E\ket{\widetilde{\psi}_z}\bra{\widetilde{\psi}_z}\right)^{\otimes (k-1)}}{\braket{\widetilde{\psi}_z|\widetilde{\psi}_z}^{k-1}}\right)\right\}\right\vert\\&\leq \Bigg\Vert\otimes_{l\neq l'}\ket{j^{(l)}}\bra{i^{(l)}}\Bigg\Vert_2\left\Vert \frac{\left(\text{Tr}_E\ket{\widetilde{\psi}_z}\bra{\widetilde{\psi}_z}\right)^{\otimes (k-1)}}{\braket{\widetilde{\psi}_z|\widetilde{\psi}_z}^{k-1}}\right\Vert_2=1\,,
\end{split}
\end{equation}
and 
\begin{equation}
    \begin{split}
        |c_z|&=\left\vert\text{Tr}\left\{\left(\otimes_{l\neq l'}\ket{j^{(l)}}\bra{i^{(l)}}\right).\left(\frac{\left(\text{Tr}_E\ket{\widetilde{\psi}_z}\bra{\widetilde{\psi}_z}\right)^{\otimes k}}{\braket{\widetilde{\psi}_z|\widetilde{\psi}_z}^{k}}\right)\right\}\right\vert\\&\leq \Bigg\Vert\otimes_{l\neq l'}\ket{j^{(l)}}\bra{i^{(l)}}\Bigg\Vert_2\left\Vert \frac{\left(\text{Tr}_E\ket{\widetilde{\psi}_z}\bra{\widetilde{\psi}_z}\right)^{\otimes k}}{\braket{\widetilde{\psi}_z|\widetilde{\psi}_z}^{k}}\right\Vert_2=1\,,
    \end{split}
\end{equation}
Eq.~\eqref{eq:lemma2-ineq2} is bounded by
\begin{equation}
   \begin{split}
        \left\Vert\frac{d}{d\Vec{u}}f_{ij}(\ket{\Psi})\right\Vert_2 \leq &~ \left(\sum_{l'=1,p'=1}^k\Vec{u}^T.\left(|M^\dagger_{p'}M_{l'}|\otimes\mathbb{I}_{S}\otimes\mathbb{I}_E\right).~\Vec{u}\right)^{1/2}+~2(k-1)\left(\Vec{u}^T.\Vec{u}\right)\\
        \leq &~\left(\sum_{l'=1,p'=1}^k\left\Vert M^\dagger_{p'}M_{l'}\right\Vert_\infty\right)^{1/2}+~2(k-1)\\
        \leq &~ 2k+2(k-1)=2(2k-1)\,,
   \end{split}
\end{equation}
where the bound in the last line arises from $\left\Vert M^\dagger_{p'}M_{l'}\right\Vert_\infty\leq\left\Vert M^\dagger_{p'}\right\Vert_\infty\Big\Vert M_{l'}\Big\Vert_\infty\leq 4$. This concludes the proof.
\\\\
Armed with the two lemmas, we now prove Theorem 1.  We apply Levy's lemma to the function $f_{ij}$ defined in \eqref{eq:fij}:
	\begin{equation}
		\begin{split}			\text{Prob}_{\psi\sim\text{Haar}(d)}\left[\vert f_{ij}(\psi)-\mathbb{E}_{\phi\sim\text{Haar}(d)}\left[f_{ij}(\phi)\right]\vert\geq \epsilon\right]&\leq 2\exp\left(-\frac{d\epsilon^2}{18\pi^3(2k-1)^2}\right)\,.
		\end{split}
	\end{equation}
    Rescaling $\epsilon\rightarrow\epsilon/d_A^{2k}$, along with performing a union bound gives us 
    \begin{equation}
		\begin{split}		\text{Prob}_{\psi\sim\text{Haar}(d)}\left[\vert f_{ij}(\psi)-\mathbb{E}_{\phi\sim\text{Haar}(d)}\left[f_{ij}(\phi)\right]\vert\geq \frac{\epsilon}{d_A^{2k}}, \text{for any $i,j$}\right]&\leq 2\exp\left(-\frac{d\epsilon^2}{18\pi^3(2k-1)^2D_R^{4k}}\right)\,.
		\end{split}
	\end{equation}
    By summing over $f_{ij}$, the above bound can be recast as
\begin{equation}
		\begin{split}
\text{Prob}_{\psi\sim\text{Haar}(d)}\left[\left\Vert\sum_z\frac{\left(\text{Tr}_E\left(\ket{\widetilde{\psi}_z}\bra{\widetilde{\psi}_z}\right)\right)^{\otimes k}}{\braket{\widetilde{\psi}_z|\widetilde{\psi}_z}^{k-1}}-\rho^{(k)}_{{\rm{gHSe}}_R}\right\Vert_{\text{entrywise},1}\geq \epsilon\right]\leq 2D_R^{2k}\exp\left(-\frac{D_SD_E\epsilon^2}{18\pi^3(2k-1)^2D_R^{4k-1}}\right)\,,
			\end{split}
	\end{equation}
    where $\rho^{(k)}_{{\rm{gHSe}}_R}=\mathbb{E}_{\phi\sim\text{Haar}(d)}\left[\left(\text{Tr}_E\ket{\phi}\bra{\phi}\right)^{\otimes k}\right]$. Using $\Vert.\Vert_{\text{entrywise},1}\geq \Vert.\Vert_1$, we get
	\begin{equation}
		\begin{split}
\text{Prob}_{\psi\sim\text{Haar}(d)}\left[\left\Vert\sum_z\frac{\left(\text{Tr}_E\left(\ket{\widetilde{\psi}_z}\bra{\widetilde{\psi}_z}\right)\right)^{\otimes k}}{\braket{\widetilde{\psi}_z|\widetilde{\psi}_z}^{k-1}}-\rho^{(k)}_{{\rm{gHSe}}_R}\right\Vert_1\geq \epsilon\right]\leq 2D_R^{2k}\exp\left(-\frac{D_SD_E\epsilon^2}{18\pi^3(2k-1)^2D_R^{4k-1}}\right)\,.
			\end{split}
	\end{equation}
	 The above bound can be rewritten as
	
	\begin{equation}
		\begin{split}
			\text{Prob}_{\psi\sim\text{Haar}(d)}\left[\left\Vert\sum_z\frac{\left(\text{Tr}_E\left(\ket{\widetilde{\psi}_z}\bra{\widetilde{\psi}_z}\right)\right)^{\otimes k}}{\braket{\widetilde{\psi}_z|\widetilde{\psi}_z}^{k-1}}-\rho^{(k)}_{{\rm{gHSe}}_R}\right\Vert_1\leq \epsilon\right]\geq 1-2D_R^{2k}\exp\left(-\frac{D_SD_E\epsilon^2}{18\pi^3(2k-1)^2D_R^{4k-1}}\right)\,.
		\end{split}
	\end{equation}
Thus, the projected ensemble forms a $\epsilon$-approximate $k$-gHSe with probability at least $1-\delta$ if
	\begin{equation}
	D_SD_E\geq\frac{18\pi^3(2k-1)^2D_R^{4k-1}}{\epsilon^2}\left(2k\log(D_R)+\log(2/\delta)\right)\,,
	\end{equation}
    which is the desired result.
\end{proof}

\twocolumngrid

\section{Emergence of the gHSe \label{app:ghse}}

\begin{figure}
\includegraphics[width=\linewidth]{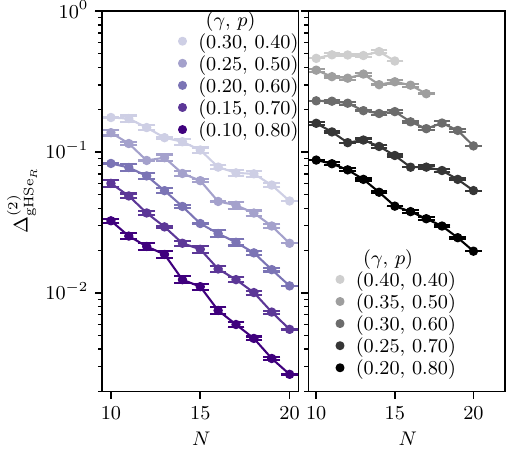}
\caption{Behaviour of $\Delta_{{\rm gHSe}_R}^{(2)}$ with system size $N$. Left panel shows the decay of $\Delta^{(2)}_{ {\rm gHSe}_ R}$ with $N$ for various points on the Holevo information transition line. The plot clearly indicates that the PPE forms a gHSe (up to the second moment) in the thermodynamic limit. Right panel shows the behaviour of $\Delta^{(2)}_{{\rm gHSe }_{R}}$ in the MVQC phase. The decay is slower and almost flat for some points in the MVQC phase, hinting at a lack of emergence of gHSe (up to the second moment) in parts of the MVQC phase in thermodynamic limit.}
\label{fig:gHSE_phase}
\end{figure}

In this Appendix we provide numerical 
evidence for the emergence -- or lack thereof -- of
the gHSe in the PPE generated by Haar-random states at the level of the second moment of the ensembles. 
To this end, we compute the trace distance between the second moments of the gHSe and the PPE
\eq{
\Delta_{{\rm gHSe}_R}^{(2)} = \frac{1}{2}\left|\left|\rho_{{\rm PPE}_R}^{(2)} - \rho_{{\rm gHSe}_R}^{(2)}\right|\right|_1\,,
\label{eq:Delta-gHSe-1}
}
at various points $(\gamma,p)$ in the phase diagram. The dependence of the distance in Eq.~\ref{eq:Delta-gHSe} with system size $N$ for different $(\gamma,p)$ is presented in Fig.~\ref{fig:gHSE_phase}. The data indicates that the PPE is well described by a gHSe, up to the second moment at least, along the Holevo information transition line (left panel of Fig.~\ref{fig:gHSE_phase}) in the thermodynamic limit.

This conclusion continues to hold true in the MIQC phase and can be argued for in a straightforward manner. Let $\Delta(\gamma,p)$ denote the distance in Eq.~\ref{eq:Delta-gHSe-1} for some $(\gamma,p)$. For any pair of points $(\gamma_1,p)$ and $(\gamma_2,p)$ such that $\gamma_2\leq\gamma_1$, we know that
\eq{\Delta(\gamma_2,p)\leq\Delta(\gamma_1,p)\,.\label{eq:contractive}}
This follows from the contractivity of the trace norm under partial traces, since for a fixed $p$, any moment of the PPE at $\gamma_2$ can be obtained by a partial trace of the corresponding moment of the PPE at $\gamma_1$. The phase diagram in Fig.~\ref{fig:phase-dia}  shows that if a point $(\gamma_1,p)$ lies on the Holevo information transition line, then the point $(\gamma_2,p)$ with $\gamma_2<\gamma_1$ necessarily lies in the MIQC phase.
Therefore, given Eq.~\ref{eq:contractive}, the emergence of the gHSe in the PPE (at the level of second moments) for any point $(\gamma_1,p)$ on the Holevo information transition line implies the same for any point $(\gamma_2,p)$ in the MIQC phase.

On the other hand, in the MVQC phase our numerical data suggests that while the gHSe emerges in the regime of large $p$ and small $\gamma$ (right panel of Fig.~\ref{fig:gHSE_phase}) it may not do so at smaller $p$. This is concomitant with the fact that the PPE is not described by a gHSe for $\gamma>0.5$ in the thermodynamic limit, which can be argued for in a straightforward manner. From the contractivity of the trace norm under partial traces, we know that if the first moment of the PPE does not asymptotically converge to that of any target ensemble for large $N$, then the higher moments do not converge to the respective moments of the target ensemble. The first moment of any gHSe is the maximally mixed state and the first moment of the PPE is simply the reduced density matrix over $R$, which is known to not be asymptotically equal to a maximally mixed state for $\gamma>0.5$. Thus, the PPE cannot be described by any gHSe for $\gamma>0.5$.

\bibliography{refs}

\begin{thebibliography}{57}%
\makeatletter
\providecommand \@ifxundefined [1]{%
 \@ifx{#1\undefined}
}%
\providecommand \@ifnum [1]{%
 \ifnum #1\expandafter \@firstoftwo
 \else \expandafter \@secondoftwo
 \fi
}%
\providecommand \@ifx [1]{%
 \ifx #1\expandafter \@firstoftwo
 \else \expandafter \@secondoftwo
 \fi
}%
\providecommand \natexlab [1]{#1}%
\providecommand \enquote  [1]{``#1''}%
\providecommand \bibnamefont  [1]{#1}%
\providecommand \bibfnamefont [1]{#1}%
\providecommand \citenamefont [1]{#1}%
\providecommand \href@noop [0]{\@secondoftwo}%
\providecommand \href [0]{\begingroup \@sanitize@url \@href}%
\providecommand \@href[1]{\@@startlink{#1}\@@href}%
\providecommand \@@href[1]{\endgroup#1\@@endlink}%
\providecommand \@sanitize@url [0]{\catcode `\\12\catcode `\$12\catcode
  `\&12\catcode `\#12\catcode `\^12\catcode `\_12\catcode `\%12\relax}%
\providecommand \@@startlink[1]{}%
\providecommand \@@endlink[0]{}%
\providecommand \url  [0]{\begingroup\@sanitize@url \@url }%
\providecommand \@url [1]{\endgroup\@href {#1}{\urlprefix }}%
\providecommand \urlprefix  [0]{URL }%
\providecommand \Eprint [0]{\href }%
\providecommand \doibase [0]{https://doi.org/}%
\providecommand \selectlanguage [0]{\@gobble}%
\providecommand \bibinfo  [0]{\@secondoftwo}%
\providecommand \bibfield  [0]{\@secondoftwo}%
\providecommand \translation [1]{[#1]}%
\providecommand \BibitemOpen [0]{}%
\providecommand \bibitemStop [0]{}%
\providecommand \bibitemNoStop [0]{.\EOS\space}%
\providecommand \EOS [0]{\spacefactor3000\relax}%
\providecommand \BibitemShut  [1]{\csname bibitem#1\endcsname}%
\let\auto@bib@innerbib\@empty
\bibitem [{\citenamefont {Nielsen}\ and\ \citenamefont
  {Chuang}(2010)}]{nielsen-chaung-book}%
  \BibitemOpen
  \bibfield  {author} {\bibinfo {author} {\bibfnamefont {M.~A.}\ \bibnamefont
  {Nielsen}}\ and\ \bibinfo {author} {\bibfnamefont {I.~L.}\ \bibnamefont
  {Chuang}},\ }\href@noop {} {\emph {\bibinfo {title} {Quantum Computation and
  Quantum Information: 10th Anniversary Edition}}}\ (\bibinfo  {publisher}
  {Cambridge University Press},\ \bibinfo {year} {2010})\ Chap.~\bibinfo
  {chapter} {9}\BibitemShut {NoStop}%
\bibitem [{\citenamefont {Nandkishore}\ and\ \citenamefont
  {Huse}(2015)}]{nandkishore2015many}%
  \BibitemOpen
  \bibfield  {author} {\bibinfo {author} {\bibfnamefont {R.}~\bibnamefont
  {Nandkishore}}\ and\ \bibinfo {author} {\bibfnamefont {D.~A.}\ \bibnamefont
  {Huse}},\ }\bibfield  {title} {\bibinfo {title} {Many-body localization and
  thermalization in quantum statistical mechanics},\ }\href
  {https://doi.org/10.1146/annurev-conmatphys-031214-014726} {\bibfield
  {journal} {\bibinfo  {journal} {Annu. Rev. Condens. Matter Phys.}\ }\textbf
  {\bibinfo {volume} {6}},\ \bibinfo {pages} {15} (\bibinfo {year}
  {2015})}\BibitemShut {NoStop}%
\bibitem [{\citenamefont {D’Alessio}\ \emph {et~al.}(2016)\citenamefont
  {D’Alessio}, \citenamefont {Kafri}, \citenamefont {Polkovnikov},\ and\
  \citenamefont {Rigol}}]{dalessio2016from}%
  \BibitemOpen
  \bibfield  {author} {\bibinfo {author} {\bibfnamefont {L.}~\bibnamefont
  {D’Alessio}}, \bibinfo {author} {\bibfnamefont {Y.}~\bibnamefont {Kafri}},
  \bibinfo {author} {\bibfnamefont {A.}~\bibnamefont {Polkovnikov}},\ and\
  \bibinfo {author} {\bibfnamefont {M.}~\bibnamefont {Rigol}},\ }\bibfield
  {title} {\bibinfo {title} {From quantum chaos and eigenstate thermalization
  to statistical mechanics and thermodynamics},\ }\href
  {https://doi.org/10.1080/00018732.2016.1198134} {\bibfield  {journal}
  {\bibinfo  {journal} {Advances in Physics}\ }\textbf {\bibinfo {volume}
  {65}},\ \bibinfo {pages} {239–362} (\bibinfo {year} {2016})}\BibitemShut
  {NoStop}%
\bibitem [{\citenamefont {Nahum}\ \emph {et~al.}(2017)\citenamefont {Nahum},
  \citenamefont {Ruhman}, \citenamefont {Vijay},\ and\ \citenamefont
  {Haah}}]{nahum2017quantum}%
  \BibitemOpen
  \bibfield  {author} {\bibinfo {author} {\bibfnamefont {A.}~\bibnamefont
  {Nahum}}, \bibinfo {author} {\bibfnamefont {J.}~\bibnamefont {Ruhman}},
  \bibinfo {author} {\bibfnamefont {S.}~\bibnamefont {Vijay}},\ and\ \bibinfo
  {author} {\bibfnamefont {J.}~\bibnamefont {Haah}},\ }\bibfield  {title}
  {\bibinfo {title} {Quantum entanglement growth under random unitary
  dynamics},\ }\href {https://doi.org/10.1103/PhysRevX.7.031016} {\bibfield
  {journal} {\bibinfo  {journal} {Phys. Rev. X}\ }\textbf {\bibinfo {volume}
  {7}},\ \bibinfo {pages} {031016} (\bibinfo {year} {2017})}\BibitemShut
  {NoStop}%
\bibitem [{\citenamefont {Garrison}\ and\ \citenamefont
  {Grover}(2018)}]{garrison2018does}%
  \BibitemOpen
  \bibfield  {author} {\bibinfo {author} {\bibfnamefont {J.~R.}\ \bibnamefont
  {Garrison}}\ and\ \bibinfo {author} {\bibfnamefont {T.}~\bibnamefont
  {Grover}},\ }\bibfield  {title} {\bibinfo {title} {Does a single eigenstate
  encode the full hamiltonian?},\ }\href
  {https://doi.org/10.1103/PhysRevX.8.021026} {\bibfield  {journal} {\bibinfo
  {journal} {Phys. Rev. X}\ }\textbf {\bibinfo {volume} {8}},\ \bibinfo {pages}
  {021026} (\bibinfo {year} {2018})}\BibitemShut {NoStop}%
\bibitem [{\citenamefont {Hayden}\ and\ \citenamefont
  {Preskill}(2007)}]{hayden2007black}%
  \BibitemOpen
  \bibfield  {author} {\bibinfo {author} {\bibfnamefont {P.}~\bibnamefont
  {Hayden}}\ and\ \bibinfo {author} {\bibfnamefont {J.}~\bibnamefont
  {Preskill}},\ }\bibfield  {title} {\bibinfo {title} {Black holes as mirrors:
  quantum information in random subsystems},\ }\href
  {https://doi.org/10.1088/1126-6708/2007/09/120} {\bibfield  {journal}
  {\bibinfo  {journal} {Journal of High Energy Physics}\ }\textbf {\bibinfo
  {volume} {2007}},\ \bibinfo {pages} {120–120} (\bibinfo {year}
  {2007})}\BibitemShut {NoStop}%
\bibitem [{\citenamefont {Sekino}\ and\ \citenamefont
  {Susskind}(2008)}]{sekino2008fast}%
  \BibitemOpen
  \bibfield  {author} {\bibinfo {author} {\bibfnamefont {Y.}~\bibnamefont
  {Sekino}}\ and\ \bibinfo {author} {\bibfnamefont {L.}~\bibnamefont
  {Susskind}},\ }\bibfield  {title} {\bibinfo {title} {Fast scramblers},\
  }\href {https://doi.org/10.1088/1126-6708/2008/10/065} {\bibfield  {journal}
  {\bibinfo  {journal} {Journal of High Energy Physics}\ }\textbf {\bibinfo
  {volume} {2008}},\ \bibinfo {pages} {065} (\bibinfo {year}
  {2008})}\BibitemShut {NoStop}%
\bibitem [{\citenamefont {Liu}\ and\ \citenamefont
  {Suh}(2014{\natexlab{a}})}]{liu2014entanglement}%
  \BibitemOpen
  \bibfield  {author} {\bibinfo {author} {\bibfnamefont {H.}~\bibnamefont
  {Liu}}\ and\ \bibinfo {author} {\bibfnamefont {S.~J.}\ \bibnamefont {Suh}},\
  }\bibfield  {title} {\bibinfo {title} {Entanglement tsunami: Universal
  scaling in holographic thermalization},\ }\href
  {https://doi.org/10.1103/PhysRevLett.112.011601} {\bibfield  {journal}
  {\bibinfo  {journal} {Phys. Rev. Lett.}\ }\textbf {\bibinfo {volume} {112}},\
  \bibinfo {pages} {011601} (\bibinfo {year} {2014}{\natexlab{a}})}\BibitemShut
  {NoStop}%
\bibitem [{\citenamefont {Liu}\ and\ \citenamefont
  {Suh}(2014{\natexlab{b}})}]{liu2014entanglement1}%
  \BibitemOpen
  \bibfield  {author} {\bibinfo {author} {\bibfnamefont {H.}~\bibnamefont
  {Liu}}\ and\ \bibinfo {author} {\bibfnamefont {S.~J.}\ \bibnamefont {Suh}},\
  }\bibfield  {title} {\bibinfo {title} {Entanglement growth during
  thermalization in holographic systems},\ }\href
  {https://doi.org/10.1103/PhysRevD.89.066012} {\bibfield  {journal} {\bibinfo
  {journal} {Phys. Rev. D}\ }\textbf {\bibinfo {volume} {89}},\ \bibinfo
  {pages} {066012} (\bibinfo {year} {2014}{\natexlab{b}})}\BibitemShut
  {NoStop}%
\bibitem [{\citenamefont {Hosur}\ \emph {et~al.}(2016)\citenamefont {Hosur},
  \citenamefont {Qi}, \citenamefont {Roberts},\ and\ \citenamefont
  {Yoshida}}]{hosur2016chaos}%
  \BibitemOpen
  \bibfield  {author} {\bibinfo {author} {\bibfnamefont {P.}~\bibnamefont
  {Hosur}}, \bibinfo {author} {\bibfnamefont {X.-L.}\ \bibnamefont {Qi}},
  \bibinfo {author} {\bibfnamefont {D.~A.}\ \bibnamefont {Roberts}},\ and\
  \bibinfo {author} {\bibfnamefont {B.}~\bibnamefont {Yoshida}},\ }\bibfield
  {title} {\bibinfo {title} {Chaos in quantum channels},\ }\href
  {https://doi.org/https://doi.org/10.1007/JHEP02(2016)004} {\bibfield
  {journal} {\bibinfo  {journal} {Journal of High Energy Physics}\ }\textbf
  {\bibinfo {volume} {2016}},\ \bibinfo {pages} {1} (\bibinfo {year}
  {2016})}\BibitemShut {NoStop}%
\bibitem [{\citenamefont {Xu}\ and\ \citenamefont
  {Swingle}(2019)}]{xu2019locality}%
  \BibitemOpen
  \bibfield  {author} {\bibinfo {author} {\bibfnamefont {S.}~\bibnamefont
  {Xu}}\ and\ \bibinfo {author} {\bibfnamefont {B.}~\bibnamefont {Swingle}},\
  }\bibfield  {title} {\bibinfo {title} {Locality, {Q}uantum {F}luctuations,
  and {S}crambling},\ }\href {https://doi.org/10.1103/PhysRevX.9.031048}
  {\bibfield  {journal} {\bibinfo  {journal} {Phys. Rev. X}\ }\textbf {\bibinfo
  {volume} {9}},\ \bibinfo {pages} {031048} (\bibinfo {year}
  {2019})}\BibitemShut {NoStop}%
\bibitem [{\citenamefont {Schuch}\ \emph {et~al.}(2008)\citenamefont {Schuch},
  \citenamefont {Wolf}, \citenamefont {Verstraete},\ and\ \citenamefont
  {Cirac}}]{schuch2008entropy}%
  \BibitemOpen
  \bibfield  {author} {\bibinfo {author} {\bibfnamefont {N.}~\bibnamefont
  {Schuch}}, \bibinfo {author} {\bibfnamefont {M.~M.}\ \bibnamefont {Wolf}},
  \bibinfo {author} {\bibfnamefont {F.}~\bibnamefont {Verstraete}},\ and\
  \bibinfo {author} {\bibfnamefont {J.~I.}\ \bibnamefont {Cirac}},\ }\bibfield
  {title} {\bibinfo {title} {Entropy scaling and simulability by matrix product
  states},\ }\href {https://doi.org/10.1103/PhysRevLett.100.030504} {\bibfield
  {journal} {\bibinfo  {journal} {Phys. Rev. Lett.}\ }\textbf {\bibinfo
  {volume} {100}},\ \bibinfo {pages} {030504} (\bibinfo {year}
  {2008})}\BibitemShut {NoStop}%
\bibitem [{\citenamefont {Harrow}\ and\ \citenamefont
  {Montanaro}(2017)}]{harrow2017nature}%
  \BibitemOpen
  \bibfield  {author} {\bibinfo {author} {\bibfnamefont {A.~W.}\ \bibnamefont
  {Harrow}}\ and\ \bibinfo {author} {\bibfnamefont {A.}~\bibnamefont
  {Montanaro}},\ }\bibfield  {title} {\bibinfo {title} {Quantum computational
  supremacy},\ }\href {https://doi.org/10.1038/nature23458} {\bibfield
  {journal} {\bibinfo  {journal} {Nature}\ }\textbf {\bibinfo {volume} {549}},\
  \bibinfo {pages} {203–209} (\bibinfo {year} {2017})}\BibitemShut {NoStop}%
\bibitem [{\citenamefont {Boixo}\ \emph {et~al.}(2018)\citenamefont {Boixo},
  \citenamefont {Isakov}, \citenamefont {Smelyanskiy}, \citenamefont {Babbush},
  \citenamefont {Ding}, \citenamefont {Jiang}, \citenamefont {Bremner},
  \citenamefont {Martinis},\ and\ \citenamefont
  {Neven}}]{boixo2018characterising}%
  \BibitemOpen
  \bibfield  {author} {\bibinfo {author} {\bibfnamefont {S.}~\bibnamefont
  {Boixo}}, \bibinfo {author} {\bibfnamefont {S.~V.}\ \bibnamefont {Isakov}},
  \bibinfo {author} {\bibfnamefont {V.~N.}\ \bibnamefont {Smelyanskiy}},
  \bibinfo {author} {\bibfnamefont {R.}~\bibnamefont {Babbush}}, \bibinfo
  {author} {\bibfnamefont {N.}~\bibnamefont {Ding}}, \bibinfo {author}
  {\bibfnamefont {Z.}~\bibnamefont {Jiang}}, \bibinfo {author} {\bibfnamefont
  {M.~J.}\ \bibnamefont {Bremner}}, \bibinfo {author} {\bibfnamefont {J.~M.}\
  \bibnamefont {Martinis}},\ and\ \bibinfo {author} {\bibfnamefont
  {H.}~\bibnamefont {Neven}},\ }\bibfield  {title} {\bibinfo {title}
  {Characterizing quantum supremacy in near-term devices},\ }\href
  {https://doi.org/10.1038/s41567-018-0124-x} {\bibfield  {journal} {\bibinfo
  {journal} {Nature Physics}\ }\textbf {\bibinfo {volume} {14}},\ \bibinfo
  {pages} {595–600} (\bibinfo {year} {2018})}\BibitemShut {NoStop}%
\bibitem [{\citenamefont {Arute}\ \emph {et~al.}(2019)\citenamefont {Arute},
  \citenamefont {Arya},\ and\ \citenamefont {Babbush~{\it et
  al.}}}]{arute2019quantum}%
  \BibitemOpen
  \bibfield  {author} {\bibinfo {author} {\bibfnamefont {F.}~\bibnamefont
  {Arute}}, \bibinfo {author} {\bibfnamefont {K.}~\bibnamefont {Arya}},\ and\
  \bibinfo {author} {\bibfnamefont {R.}~\bibnamefont {Babbush~{\it et al.}}},\
  }\bibfield  {title} {\bibinfo {title} {Quantum supremacy using a programmable
  superconducting processor},\ }\href
  {https://doi.org/10.1038/s41586-019-1666-5} {\bibfield  {journal} {\bibinfo
  {journal} {Nature}\ }\textbf {\bibinfo {volume} {574}},\ \bibinfo {pages}
  {505–510} (\bibinfo {year} {2019})}\BibitemShut {NoStop}%
\bibitem [{\citenamefont {Vidal}\ and\ \citenamefont
  {Werner}(2002)}]{vidal2002computable}%
  \BibitemOpen
  \bibfield  {author} {\bibinfo {author} {\bibfnamefont {G.}~\bibnamefont
  {Vidal}}\ and\ \bibinfo {author} {\bibfnamefont {R.~F.}\ \bibnamefont
  {Werner}},\ }\bibfield  {title} {\bibinfo {title} {Computable measure of
  entanglement},\ }\href {https://doi.org/10.1103/PhysRevA.65.032314}
  {\bibfield  {journal} {\bibinfo  {journal} {Phys. Rev. A}\ }\textbf {\bibinfo
  {volume} {65}},\ \bibinfo {pages} {032314} (\bibinfo {year}
  {2002})}\BibitemShut {NoStop}%
\bibitem [{\citenamefont {Plenio}(2005)}]{plenio2005logarithmic}%
  \BibitemOpen
  \bibfield  {author} {\bibinfo {author} {\bibfnamefont {M.~B.}\ \bibnamefont
  {Plenio}},\ }\bibfield  {title} {\bibinfo {title} {Logarithmic negativity: A
  full entanglement monotone that is not convex},\ }\href
  {https://doi.org/10.1103/PhysRevLett.95.090503} {\bibfield  {journal}
  {\bibinfo  {journal} {Phys. Rev. Lett.}\ }\textbf {\bibinfo {volume} {95}},\
  \bibinfo {pages} {090503} (\bibinfo {year} {2005})}\BibitemShut {NoStop}%
\bibitem [{\citenamefont {Eisert}(2006)}]{eisert2006entanglement}%
  \BibitemOpen
  \bibfield  {author} {\bibinfo {author} {\bibfnamefont {J.}~\bibnamefont
  {Eisert}},\ }\href {https://arxiv.org/abs/quant-ph/0610253} {\bibinfo {title}
  {Entanglement in quantum information theory}} (\bibinfo {year} {2006}),\
  \Eprint {https://arxiv.org/abs/quant-ph/0610253} {arXiv:quant-ph/0610253
  [quant-ph]} \BibitemShut {NoStop}%
\bibitem [{\citenamefont {Choi}\ \emph {et~al.}(2023)\citenamefont {Choi},
  \citenamefont {Shaw}, \citenamefont {Madjarov}, \citenamefont {Xie},
  \citenamefont {Finkelstein}, \citenamefont {Covey}, \citenamefont {Cotler},
  \citenamefont {Mark}, \citenamefont {Huang}, \citenamefont {Kale},
  \citenamefont {Pichler}, \citenamefont {Brand{\~a}o}, \citenamefont {Choi},\
  and\ \citenamefont {Endres}}]{choi2023preparing}%
  \BibitemOpen
  \bibfield  {author} {\bibinfo {author} {\bibfnamefont {J.}~\bibnamefont
  {Choi}}, \bibinfo {author} {\bibfnamefont {A.~L.}\ \bibnamefont {Shaw}},
  \bibinfo {author} {\bibfnamefont {I.~S.}\ \bibnamefont {Madjarov}}, \bibinfo
  {author} {\bibfnamefont {X.}~\bibnamefont {Xie}}, \bibinfo {author}
  {\bibfnamefont {R.}~\bibnamefont {Finkelstein}}, \bibinfo {author}
  {\bibfnamefont {J.~P.}\ \bibnamefont {Covey}}, \bibinfo {author}
  {\bibfnamefont {J.~S.}\ \bibnamefont {Cotler}}, \bibinfo {author}
  {\bibfnamefont {D.~K.}\ \bibnamefont {Mark}}, \bibinfo {author}
  {\bibfnamefont {H.-Y.}\ \bibnamefont {Huang}}, \bibinfo {author}
  {\bibfnamefont {A.}~\bibnamefont {Kale}}, \bibinfo {author} {\bibfnamefont
  {H.}~\bibnamefont {Pichler}}, \bibinfo {author} {\bibfnamefont {F.~G. S.~L.}\
  \bibnamefont {Brand{\~a}o}}, \bibinfo {author} {\bibfnamefont
  {S.}~\bibnamefont {Choi}},\ and\ \bibinfo {author} {\bibfnamefont
  {M.}~\bibnamefont {Endres}},\ }\bibfield  {title} {\bibinfo {title}
  {Preparing random states and benchmarking with many-body quantum chaos},\
  }\href {https://doi.org/10.1038/s41586-022-05442-1} {\bibfield  {journal}
  {\bibinfo  {journal} {Nature}\ }\textbf {\bibinfo {volume} {613}},\ \bibinfo
  {pages} {468} (\bibinfo {year} {2023})}\BibitemShut {NoStop}%
\bibitem [{\citenamefont {Cotler}\ \emph {et~al.}(2023)\citenamefont {Cotler},
  \citenamefont {Mark}, \citenamefont {Huang}, \citenamefont {Hern\'andez},
  \citenamefont {Choi}, \citenamefont {Shaw}, \citenamefont {Endres},\ and\
  \citenamefont {Choi}}]{cotler2023emergent}%
  \BibitemOpen
  \bibfield  {author} {\bibinfo {author} {\bibfnamefont {J.~S.}\ \bibnamefont
  {Cotler}}, \bibinfo {author} {\bibfnamefont {D.~K.}\ \bibnamefont {Mark}},
  \bibinfo {author} {\bibfnamefont {H.-Y.}\ \bibnamefont {Huang}}, \bibinfo
  {author} {\bibfnamefont {F.}~\bibnamefont {Hern\'andez}}, \bibinfo {author}
  {\bibfnamefont {J.}~\bibnamefont {Choi}}, \bibinfo {author} {\bibfnamefont
  {A.~L.}\ \bibnamefont {Shaw}}, \bibinfo {author} {\bibfnamefont
  {M.}~\bibnamefont {Endres}},\ and\ \bibinfo {author} {\bibfnamefont
  {S.}~\bibnamefont {Choi}},\ }\bibfield  {title} {\bibinfo {title} {{Emergent
  Quantum State Designs from Individual Many-Body Wave Functions}},\ }\href
  {https://doi.org/10.1103/PRXQuantum.4.010311} {\bibfield  {journal} {\bibinfo
   {journal} {PRX Quantum}\ }\textbf {\bibinfo {volume} {4}},\ \bibinfo {pages}
  {010311} (\bibinfo {year} {2023})}\BibitemShut {NoStop}%
\bibitem [{\citenamefont {Ho}\ and\ \citenamefont {Choi}(2022)}]{ho2022exact}%
  \BibitemOpen
  \bibfield  {author} {\bibinfo {author} {\bibfnamefont {W.~W.}\ \bibnamefont
  {Ho}}\ and\ \bibinfo {author} {\bibfnamefont {S.}~\bibnamefont {Choi}},\
  }\bibfield  {title} {\bibinfo {title} {Exact emergent quantum state designs
  from quantum chaotic dynamics},\ }\href
  {https://doi.org/10.1103/PhysRevLett.128.060601} {\bibfield  {journal}
  {\bibinfo  {journal} {Phys. Rev. Lett.}\ }\textbf {\bibinfo {volume} {128}},\
  \bibinfo {pages} {060601} (\bibinfo {year} {2022})}\BibitemShut {NoStop}%
\bibitem [{\citenamefont {Ippoliti}\ and\ \citenamefont
  {Ho}(2022)}]{ippoliti2022solvablemodelofdeep}%
  \BibitemOpen
  \bibfield  {author} {\bibinfo {author} {\bibfnamefont {M.}~\bibnamefont
  {Ippoliti}}\ and\ \bibinfo {author} {\bibfnamefont {W.~W.}\ \bibnamefont
  {Ho}},\ }\bibfield  {title} {\bibinfo {title} {Solvable model of deep
  thermalization with distinct design times},\ }\href
  {https://doi.org/10.22331/q-2022-12-29-886} {\bibfield  {journal} {\bibinfo
  {journal} {{Quantum}}\ }\textbf {\bibinfo {volume} {6}},\ \bibinfo {pages}
  {886} (\bibinfo {year} {2022})}\BibitemShut {NoStop}%
\bibitem [{\citenamefont {Lucas}\ \emph {et~al.}(2023)\citenamefont {Lucas},
  \citenamefont {Piroli}, \citenamefont {De~Nardis},\ and\ \citenamefont
  {De~Luca}}]{lucas2023freefermiondeepthm}%
  \BibitemOpen
  \bibfield  {author} {\bibinfo {author} {\bibfnamefont {M.}~\bibnamefont
  {Lucas}}, \bibinfo {author} {\bibfnamefont {L.}~\bibnamefont {Piroli}},
  \bibinfo {author} {\bibfnamefont {J.}~\bibnamefont {De~Nardis}},\ and\
  \bibinfo {author} {\bibfnamefont {A.}~\bibnamefont {De~Luca}},\ }\bibfield
  {title} {\bibinfo {title} {Generalized deep thermalization for free
  fermions},\ }\href {https://doi.org/10.1103/PhysRevA.107.032215} {\bibfield
  {journal} {\bibinfo  {journal} {Phys. Rev. A}\ }\textbf {\bibinfo {volume}
  {107}},\ \bibinfo {pages} {032215} (\bibinfo {year} {2023})}\BibitemShut
  {NoStop}%
\bibitem [{\citenamefont {Ippoliti}\ and\ \citenamefont
  {Ho}(2023)}]{ippolitu2023dynamical}%
  \BibitemOpen
  \bibfield  {author} {\bibinfo {author} {\bibfnamefont {M.}~\bibnamefont
  {Ippoliti}}\ and\ \bibinfo {author} {\bibfnamefont {W.~W.}\ \bibnamefont
  {Ho}},\ }\bibfield  {title} {\bibinfo {title} {Dynamical purification and the
  emergence of quantum state designs from the projected ensemble},\ }\href
  {https://doi.org/10.1103/PRXQuantum.4.030322} {\bibfield  {journal} {\bibinfo
   {journal} {PRX Quantum}\ }\textbf {\bibinfo {volume} {4}},\ \bibinfo {pages}
  {030322} (\bibinfo {year} {2023})}\BibitemShut {NoStop}%
\bibitem [{\citenamefont {Bhore}\ \emph {et~al.}(2023)\citenamefont {Bhore},
  \citenamefont {Desaules},\ and\ \citenamefont {Papi\ifmmode~\acute{c}\else
  \'{c}\fi{}}}]{bhore2023deepthmconstrained}%
  \BibitemOpen
  \bibfield  {author} {\bibinfo {author} {\bibfnamefont {T.}~\bibnamefont
  {Bhore}}, \bibinfo {author} {\bibfnamefont {J.-Y.}\ \bibnamefont
  {Desaules}},\ and\ \bibinfo {author} {\bibfnamefont {Z.}~\bibnamefont
  {Papi\ifmmode~\acute{c}\else \'{c}\fi{}}},\ }\bibfield  {title} {\bibinfo
  {title} {Deep thermalization in constrained quantum systems},\ }\href
  {https://doi.org/10.1103/PhysRevB.108.104317} {\bibfield  {journal} {\bibinfo
   {journal} {Phys. Rev. B}\ }\textbf {\bibinfo {volume} {108}},\ \bibinfo
  {pages} {104317} (\bibinfo {year} {2023})}\BibitemShut {NoStop}%
\bibitem [{\citenamefont {Varikuti}\ and\ \citenamefont
  {Bandyopadhyay}(2024)}]{varikuti2024unravelingemergence}%
  \BibitemOpen
  \bibfield  {author} {\bibinfo {author} {\bibfnamefont {N.~D.}\ \bibnamefont
  {Varikuti}}\ and\ \bibinfo {author} {\bibfnamefont {S.}~\bibnamefont
  {Bandyopadhyay}},\ }\bibfield  {title} {\bibinfo {title} {Unraveling the
  emergence of quantum state designs in systems with symmetry},\ }\href
  {https://doi.org/10.22331/q-2024-08-29-1456} {\bibfield  {journal} {\bibinfo
  {journal} {{Quantum}}\ }\textbf {\bibinfo {volume} {8}},\ \bibinfo {pages}
  {1456} (\bibinfo {year} {2024})}\BibitemShut {NoStop}%
\bibitem [{\citenamefont {Chan}\ and\ \citenamefont
  {De~Luca}(2024)}]{chan2024pe}%
  \BibitemOpen
  \bibfield  {author} {\bibinfo {author} {\bibfnamefont {A.}~\bibnamefont
  {Chan}}\ and\ \bibinfo {author} {\bibfnamefont {A.}~\bibnamefont {De~Luca}},\
  }\bibfield  {title} {\bibinfo {title} {Projected state ensemble of a generic
  model of many-body quantum chaos},\ }\href
  {https://doi.org/10.1088/1751-8121/ad7211} {\bibfield  {journal} {\bibinfo
  {journal} {Journal of Physics A: Mathematical and Theoretical}\ }\textbf
  {\bibinfo {volume} {57}},\ \bibinfo {pages} {405001} (\bibinfo {year}
  {2024})}\BibitemShut {NoStop}%
\bibitem [{\citenamefont {Mark}\ \emph {et~al.}(2024)\citenamefont {Mark},
  \citenamefont {Surace}, \citenamefont {Elben}, \citenamefont {Shaw},
  \citenamefont {Choi}, \citenamefont {Refael}, \citenamefont {Endres},\ and\
  \citenamefont {Choi}}]{mark2024maximum}%
  \BibitemOpen
  \bibfield  {author} {\bibinfo {author} {\bibfnamefont {D.~K.}\ \bibnamefont
  {Mark}}, \bibinfo {author} {\bibfnamefont {F.}~\bibnamefont {Surace}},
  \bibinfo {author} {\bibfnamefont {A.}~\bibnamefont {Elben}}, \bibinfo
  {author} {\bibfnamefont {A.~L.}\ \bibnamefont {Shaw}}, \bibinfo {author}
  {\bibfnamefont {J.}~\bibnamefont {Choi}}, \bibinfo {author} {\bibfnamefont
  {G.}~\bibnamefont {Refael}}, \bibinfo {author} {\bibfnamefont
  {M.}~\bibnamefont {Endres}},\ and\ \bibinfo {author} {\bibfnamefont
  {S.}~\bibnamefont {Choi}},\ }\bibfield  {title} {\bibinfo {title} {{Maximum
  Entropy Principle in Deep Thermalization and in Hilbert-Space Ergodicity}},\
  }\href {https://doi.org/10.1103/PhysRevX.14.041051} {\bibfield  {journal}
  {\bibinfo  {journal} {Phys. Rev. X}\ }\textbf {\bibinfo {volume} {14}},\
  \bibinfo {pages} {041051} (\bibinfo {year} {2024})}\BibitemShut {NoStop}%
\bibitem [{\citenamefont {Manna}\ \emph {et~al.}(2025)\citenamefont {Manna},
  \citenamefont {Roy},\ and\ \citenamefont {Sreejith}}]{manna2025peconserved}%
  \BibitemOpen
  \bibfield  {author} {\bibinfo {author} {\bibfnamefont {S.}~\bibnamefont
  {Manna}}, \bibinfo {author} {\bibfnamefont {S.}~\bibnamefont {Roy}},\ and\
  \bibinfo {author} {\bibfnamefont {G.~J.}\ \bibnamefont {Sreejith}},\
  }\bibfield  {title} {\bibinfo {title} {Projected ensemble in a system with
  locally supported conserved charges},\ }\href
  {https://doi.org/10.1103/PhysRevB.111.144302} {\bibfield  {journal} {\bibinfo
   {journal} {Phys. Rev. B}\ }\textbf {\bibinfo {volume} {111}},\ \bibinfo
  {pages} {144302} (\bibinfo {year} {2025})}\BibitemShut {NoStop}%
\bibitem [{\citenamefont {Yu}\ \emph {et~al.}(2025)\citenamefont {Yu},
  \citenamefont {Ho},\ and\ \citenamefont
  {Kos}}]{yu2025mixedstatedeepthermalization}%
  \BibitemOpen
  \bibfield  {author} {\bibinfo {author} {\bibfnamefont {X.-H.}\ \bibnamefont
  {Yu}}, \bibinfo {author} {\bibfnamefont {W.~W.}\ \bibnamefont {Ho}},\ and\
  \bibinfo {author} {\bibfnamefont {P.}~\bibnamefont {Kos}},\ }\href
  {https://arxiv.org/abs/2505.07795} {\bibinfo {title} {Mixed state deep
  thermalization}} (\bibinfo {year} {2025}),\ \Eprint
  {https://arxiv.org/abs/2505.07795} {arXiv:2505.07795 [quant-ph]} \BibitemShut
  {NoStop}%
\bibitem [{\citenamefont {Sherry}\ and\ \citenamefont
  {Roy}(2025{\natexlab{a}})}]{sherry2025mixedstatesexhibitdeep}%
  \BibitemOpen
  \bibfield  {author} {\bibinfo {author} {\bibfnamefont {A.}~\bibnamefont
  {Sherry}}\ and\ \bibinfo {author} {\bibfnamefont {S.}~\bibnamefont {Roy}},\
  }\href {https://arxiv.org/abs/2507.14135} {\bibinfo {title} {Do mixed states
  exhibit deep thermalisation?}} (\bibinfo {year} {2025}{\natexlab{a}}),\
  \Eprint {https://arxiv.org/abs/2507.14135} {arXiv:2507.14135 [quant-ph]}
  \BibitemShut {NoStop}%
\bibitem [{\citenamefont {Mandal}\ \emph {et~al.}(2025)\citenamefont {Mandal},
  \citenamefont {Claeys},\ and\ \citenamefont
  {Roy}}]{mandal2025partialprojectedensemblesspatiotemporal}%
  \BibitemOpen
  \bibfield  {author} {\bibinfo {author} {\bibfnamefont {S.}~\bibnamefont
  {Mandal}}, \bibinfo {author} {\bibfnamefont {P.~W.}\ \bibnamefont {Claeys}},\
  and\ \bibinfo {author} {\bibfnamefont {S.}~\bibnamefont {Roy}},\ }\href
  {https://arxiv.org/abs/2508.05632} {\bibinfo {title} {Partial projected
  ensembles and spatiotemporal structure of information scrambling}} (\bibinfo
  {year} {2025}),\ \Eprint {https://arxiv.org/abs/2508.05632} {arXiv:2508.05632
  [quant-ph]} \BibitemShut {NoStop}%
\bibitem [{\citenamefont {Sherry}\ and\ \citenamefont
  {Roy}(2025{\natexlab{b}})}]{sherry2025miqc}%
  \BibitemOpen
  \bibfield  {author} {\bibinfo {author} {\bibfnamefont {A.}~\bibnamefont
  {Sherry}}\ and\ \bibinfo {author} {\bibfnamefont {S.}~\bibnamefont {Roy}},\
  }\bibfield  {title} {\bibinfo {title} {Measurement-invisible quantum
  correlations in scrambling dynamics},\ }\href
  {https://doi.org/10.1103/PhysRevB.111.L180301} {\bibfield  {journal}
  {\bibinfo  {journal} {Phys. Rev. B}\ }\textbf {\bibinfo {volume} {111}},\
  \bibinfo {pages} {L180301} (\bibinfo {year}
  {2025}{\natexlab{b}})}\BibitemShut {NoStop}%
\bibitem [{\citenamefont {Bejan}\ \emph {et~al.}(2025)\citenamefont {Bejan},
  \citenamefont {B\'eri},\ and\ \citenamefont {McGinley}}]{bejan2025matchgate}%
  \BibitemOpen
  \bibfield  {author} {\bibinfo {author} {\bibfnamefont {M.}~\bibnamefont
  {Bejan}}, \bibinfo {author} {\bibfnamefont {B.}~\bibnamefont {B\'eri}},\ and\
  \bibinfo {author} {\bibfnamefont {M.}~\bibnamefont {McGinley}},\ }\bibfield
  {title} {\bibinfo {title} {{Matchgate Circuits Deeply Thermalize}},\ }\href
  {https://doi.org/10.1103/v8kp-39ry} {\bibfield  {journal} {\bibinfo
  {journal} {Phys. Rev. Lett.}\ }\textbf {\bibinfo {volume} {135}},\ \bibinfo
  {pages} {020401} (\bibinfo {year} {2025})}\BibitemShut {NoStop}%
\bibitem [{\citenamefont {Deutsch}(1991)}]{deutsch1991quantum}%
  \BibitemOpen
  \bibfield  {author} {\bibinfo {author} {\bibfnamefont {J.~M.}\ \bibnamefont
  {Deutsch}},\ }\bibfield  {title} {\bibinfo {title} {Quantum statistical
  mechanics in a closed system},\ }\href
  {https://doi.org/10.1103/PhysRevA.43.2046} {\bibfield  {journal} {\bibinfo
  {journal} {Phys. Rev. A}\ }\textbf {\bibinfo {volume} {43}},\ \bibinfo
  {pages} {2046} (\bibinfo {year} {1991})}\BibitemShut {NoStop}%
\bibitem [{\citenamefont {Srednicki}(1994)}]{srednicki1994chaos}%
  \BibitemOpen
  \bibfield  {author} {\bibinfo {author} {\bibfnamefont {M.}~\bibnamefont
  {Srednicki}},\ }\bibfield  {title} {\bibinfo {title} {Chaos and quantum
  thermalization},\ }\href {https://doi.org/10.1103/PhysRevE.50.888} {\bibfield
   {journal} {\bibinfo  {journal} {Phys. Rev. E}\ }\textbf {\bibinfo {volume}
  {50}},\ \bibinfo {pages} {888} (\bibinfo {year} {1994})}\BibitemShut
  {NoStop}%
\bibitem [{\citenamefont {Rigol}\ \emph {et~al.}(2008)\citenamefont {Rigol},
  \citenamefont {Dunjko},\ and\ \citenamefont
  {Olshanii}}]{rigol2008thermalisation}%
  \BibitemOpen
  \bibfield  {author} {\bibinfo {author} {\bibfnamefont {M.}~\bibnamefont
  {Rigol}}, \bibinfo {author} {\bibfnamefont {V.}~\bibnamefont {Dunjko}},\ and\
  \bibinfo {author} {\bibfnamefont {M.}~\bibnamefont {Olshanii}},\ }\bibfield
  {title} {\bibinfo {title} {Thermalization and its mechanism for generic
  isolated quantum systems},\ }\href {https://doi.org/10.1038/nature06838}
  {\bibfield  {journal} {\bibinfo  {journal} {Nature}\ }\textbf {\bibinfo
  {volume} {452}},\ \bibinfo {pages} {854–858} (\bibinfo {year}
  {2008})}\BibitemShut {NoStop}%
\bibitem [{\citenamefont {Deutsch}(2018)}]{deutsch2018eigenstate}%
  \BibitemOpen
  \bibfield  {author} {\bibinfo {author} {\bibfnamefont {J.~M.}\ \bibnamefont
  {Deutsch}},\ }\bibfield  {title} {\bibinfo {title} {Eigenstate thermalization
  hypothesis},\ }\href
  {http://iopscience.iop.org/article/10.1088/1361-6633/aac9f1/meta} {\bibfield
  {journal} {\bibinfo  {journal} {Rep. Prog. Phys.}\ }\textbf {\bibinfo
  {volume} {81}},\ \bibinfo {pages} {082001} (\bibinfo {year}
  {2018})}\BibitemShut {NoStop}%
\bibitem [{\citenamefont {Lazarides}\ \emph {et~al.}(2014)\citenamefont
  {Lazarides}, \citenamefont {Das},\ and\ \citenamefont
  {Moessner}}]{lazarides2014equilibrium}%
  \BibitemOpen
  \bibfield  {author} {\bibinfo {author} {\bibfnamefont {A.}~\bibnamefont
  {Lazarides}}, \bibinfo {author} {\bibfnamefont {A.}~\bibnamefont {Das}},\
  and\ \bibinfo {author} {\bibfnamefont {R.}~\bibnamefont {Moessner}},\
  }\bibfield  {title} {\bibinfo {title} {Equilibrium states of generic quantum
  systems subject to periodic driving},\ }\href
  {https://doi.org/10.1103/PhysRevE.90.012110} {\bibfield  {journal} {\bibinfo
  {journal} {Phys. Rev. E}\ }\textbf {\bibinfo {volume} {90}},\ \bibinfo
  {pages} {012110} (\bibinfo {year} {2014})}\BibitemShut {NoStop}%
\bibitem [{\citenamefont {D'Alessio}\ and\ \citenamefont
  {Rigol}(2014)}]{dalessio2014long}%
  \BibitemOpen
  \bibfield  {author} {\bibinfo {author} {\bibfnamefont {L.}~\bibnamefont
  {D'Alessio}}\ and\ \bibinfo {author} {\bibfnamefont {M.}~\bibnamefont
  {Rigol}},\ }\bibfield  {title} {\bibinfo {title} {Long-time behavior of
  isolated periodically driven interacting lattice systems},\ }\href
  {https://doi.org/10.1103/PhysRevX.4.041048} {\bibfield  {journal} {\bibinfo
  {journal} {Phys. Rev. X}\ }\textbf {\bibinfo {volume} {4}},\ \bibinfo {pages}
  {041048} (\bibinfo {year} {2014})}\BibitemShut {NoStop}%
\bibitem [{\citenamefont {Jozsa}\ \emph {et~al.}(1994)\citenamefont {Jozsa},
  \citenamefont {Robb},\ and\ \citenamefont {Wootters}}]{jozsa1994accessinfo}%
  \BibitemOpen
  \bibfield  {author} {\bibinfo {author} {\bibfnamefont {R.}~\bibnamefont
  {Jozsa}}, \bibinfo {author} {\bibfnamefont {D.}~\bibnamefont {Robb}},\ and\
  \bibinfo {author} {\bibfnamefont {W.~K.}\ \bibnamefont {Wootters}},\
  }\bibfield  {title} {\bibinfo {title} {Lower bound for accessible information
  in quantum mechanics},\ }\href {https://doi.org/10.1103/PhysRevA.49.668}
  {\bibfield  {journal} {\bibinfo  {journal} {Phys. Rev. A}\ }\textbf {\bibinfo
  {volume} {49}},\ \bibinfo {pages} {668} (\bibinfo {year} {1994})}\BibitemShut
  {NoStop}%
\bibitem [{\citenamefont {Goldstein}\ \emph
  {et~al.}(2006{\natexlab{a}})\citenamefont {Goldstein}, \citenamefont
  {Lebowitz}, \citenamefont {Tumulka},\ and\ \citenamefont
  {Zangh{\`i}}}]{goldstein2006distribution}%
  \BibitemOpen
  \bibfield  {author} {\bibinfo {author} {\bibfnamefont {S.}~\bibnamefont
  {Goldstein}}, \bibinfo {author} {\bibfnamefont {J.~L.}\ \bibnamefont
  {Lebowitz}}, \bibinfo {author} {\bibfnamefont {R.}~\bibnamefont {Tumulka}},\
  and\ \bibinfo {author} {\bibfnamefont {N.}~\bibnamefont {Zangh{\`i}}},\
  }\bibfield  {title} {\bibinfo {title} {On the distribution of the wave
  function for systems in thermal equilibrium},\ }\href
  {https://doi.org/10.1007/s10955-006-9210-z} {\bibfield  {journal} {\bibinfo
  {journal} {Journal of Statistical Physics}\ }\textbf {\bibinfo {volume}
  {125}},\ \bibinfo {pages} {1193} (\bibinfo {year}
  {2006}{\natexlab{a}})}\BibitemShut {NoStop}%
\bibitem [{\citenamefont {Liu}\ \emph {et~al.}(2024)\citenamefont {Liu},
  \citenamefont {Huang},\ and\ \citenamefont {Ho}}]{chang2024gaussian}%
  \BibitemOpen
  \bibfield  {author} {\bibinfo {author} {\bibfnamefont {C.}~\bibnamefont
  {Liu}}, \bibinfo {author} {\bibfnamefont {Q.~C.}\ \bibnamefont {Huang}},\
  and\ \bibinfo {author} {\bibfnamefont {W.~W.}\ \bibnamefont {Ho}},\
  }\bibfield  {title} {\bibinfo {title} {Deep thermalization in gaussian
  continuous-variable quantum systems},\ }\href
  {https://doi.org/10.1103/PhysRevLett.133.260401} {\bibfield  {journal}
  {\bibinfo  {journal} {Phys. Rev. Lett.}\ }\textbf {\bibinfo {volume} {133}},\
  \bibinfo {pages} {260401} (\bibinfo {year} {2024})}\BibitemShut {NoStop}%
\bibitem [{\citenamefont {Chang}\ \emph {et~al.}(2025)\citenamefont {Chang},
  \citenamefont {Shrotriya}, \citenamefont {Ho},\ and\ \citenamefont
  {Ippoliti}}]{chang2025charge}%
  \BibitemOpen
  \bibfield  {author} {\bibinfo {author} {\bibfnamefont {R.-A.}\ \bibnamefont
  {Chang}}, \bibinfo {author} {\bibfnamefont {H.}~\bibnamefont {Shrotriya}},
  \bibinfo {author} {\bibfnamefont {W.~W.}\ \bibnamefont {Ho}},\ and\ \bibinfo
  {author} {\bibfnamefont {M.}~\bibnamefont {Ippoliti}},\ }\bibfield  {title}
  {\bibinfo {title} {Deep thermalization under charge-conserving quantum
  dynamics},\ }\href {https://doi.org/10.1103/PRXQuantum.6.020343} {\bibfield
  {journal} {\bibinfo  {journal} {PRX Quantum}\ }\textbf {\bibinfo {volume}
  {6}},\ \bibinfo {pages} {020343} (\bibinfo {year} {2025})}\BibitemShut
  {NoStop}%
\bibitem [{\citenamefont {Holevo}(1973)}]{holevo1973accinfo}%
  \BibitemOpen
  \bibfield  {author} {\bibinfo {author} {\bibfnamefont {A.~S.}\ \bibnamefont
  {Holevo}},\ }\bibfield  {title} {\bibinfo {title} {Bounds for the quantity of
  information transmitted by a quantum communication channel},\ }\href@noop {}
  {\bibfield  {journal} {\bibinfo  {journal} {Problems Inform. Transmission}\
  }\textbf {\bibinfo {volume} {9}},\ \bibinfo {pages} {177} (\bibinfo {year}
  {1973})}\BibitemShut {NoStop}%
\bibitem [{\citenamefont {Preskill}(2025)}]{preskill2025quantumshannontheory}%
  \BibitemOpen
  \bibfield  {author} {\bibinfo {author} {\bibfnamefont {J.}~\bibnamefont
  {Preskill}},\ }\href {https://arxiv.org/abs/1604.07450} {\bibinfo {title}
  {Quantum shannon theory}} (\bibinfo {year} {2025}),\ \Eprint
  {https://arxiv.org/abs/1604.07450} {arXiv:1604.07450 [quant-ph]} \BibitemShut
  {NoStop}%
\bibitem [{\citenamefont {Hall}(1998)}]{hall1998randomquantumcorrelations}%
  \BibitemOpen
  \bibfield  {author} {\bibinfo {author} {\bibfnamefont {M.~J.}\ \bibnamefont
  {Hall}},\ }\bibfield  {title} {\bibinfo {title} {Random quantum correlations
  and density operator distributions},\ }\href
  {https://doi.org/https://doi.org/10.1016/S0375-9601(98)00190-X} {\bibfield
  {journal} {\bibinfo  {journal} {Physics Letters A}\ }\textbf {\bibinfo
  {volume} {242}},\ \bibinfo {pages} {123} (\bibinfo {year}
  {1998})}\BibitemShut {NoStop}%
\bibitem [{\citenamefont {Zyczkowski}\ and\ \citenamefont
  {Sommers}(2001)}]{karol2001inducedmeasures}%
  \BibitemOpen
  \bibfield  {author} {\bibinfo {author} {\bibfnamefont {K.}~\bibnamefont
  {Zyczkowski}}\ and\ \bibinfo {author} {\bibfnamefont {H.-J.}\ \bibnamefont
  {Sommers}},\ }\bibfield  {title} {\bibinfo {title} {Induced measures in the
  space of mixed quantum states},\ }\href
  {https://doi.org/10.1088/0305-4470/34/35/335} {\bibfield  {journal} {\bibinfo
   {journal} {Journal of Physics A: Mathematical and General}\ }\textbf
  {\bibinfo {volume} {34}},\ \bibinfo {pages} {7111} (\bibinfo {year}
  {2001})}\BibitemShut {NoStop}%
\bibitem [{\citenamefont {Bansal}\ \emph {et~al.}(2025)\citenamefont {Bansal},
  \citenamefont {Mok}, \citenamefont {Bharti}, \citenamefont {Koh},\ and\
  \citenamefont {Haug}}]{bansal2025pseudordm}%
  \BibitemOpen
  \bibfield  {author} {\bibinfo {author} {\bibfnamefont {N.}~\bibnamefont
  {Bansal}}, \bibinfo {author} {\bibfnamefont {W.-K.}\ \bibnamefont {Mok}},
  \bibinfo {author} {\bibfnamefont {K.}~\bibnamefont {Bharti}}, \bibinfo
  {author} {\bibfnamefont {D.~E.}\ \bibnamefont {Koh}},\ and\ \bibinfo {author}
  {\bibfnamefont {T.}~\bibnamefont {Haug}},\ }\bibfield  {title} {\bibinfo
  {title} {Pseudorandom density matrices},\ }\href
  {https://doi.org/10.1103/PRXQuantum.6.020322} {\bibfield  {journal} {\bibinfo
   {journal} {PRX Quantum}\ }\textbf {\bibinfo {volume} {6}},\ \bibinfo {pages}
  {020322} (\bibinfo {year} {2025})}\BibitemShut {NoStop}%
\bibitem [{\citenamefont {Page}(1993)}]{page1993average}%
  \BibitemOpen
  \bibfield  {author} {\bibinfo {author} {\bibfnamefont {D.~N.}\ \bibnamefont
  {Page}},\ }\bibfield  {title} {\bibinfo {title} {Average entropy of a
  subsystem},\ }\href {https://doi.org/10.1103/PhysRevLett.71.1291} {\bibfield
  {journal} {\bibinfo  {journal} {Phys. Rev. Lett.}\ }\textbf {\bibinfo
  {volume} {71}},\ \bibinfo {pages} {1291} (\bibinfo {year}
  {1993})}\BibitemShut {NoStop}%
\bibitem [{\citenamefont {Forrester}(2010)}]{forrestor2010rmt}%
  \BibitemOpen
  \bibfield  {author} {\bibinfo {author} {\bibfnamefont {P.}~\bibnamefont
  {Forrester}},\ }\href {http://www.jstor.org/stable/j.ctt7t5vq} {\emph
  {\bibinfo {title} {Log-Gases and Random Matrices (LMS-34)}}}\ (\bibinfo
  {publisher} {Princeton University Press},\ \bibinfo {year}
  {2010})\BibitemShut {NoStop}%
\bibitem [{\citenamefont {Henderson}\ and\ \citenamefont
  {Vedral}(2001)}]{Henderson2001CQCorr}%
  \BibitemOpen
  \bibfield  {author} {\bibinfo {author} {\bibfnamefont {L.}~\bibnamefont
  {Henderson}}\ and\ \bibinfo {author} {\bibfnamefont {V.}~\bibnamefont
  {Vedral}},\ }\bibfield  {title} {\bibinfo {title} {Classical, quantum and
  total correlations},\ }\href {https://doi.org/10.1088/0305-4470/34/35/315}
  {\bibfield  {journal} {\bibinfo  {journal} {Journal of Physics A:
  Mathematical and General}\ }\textbf {\bibinfo {volume} {34}},\ \bibinfo
  {pages} {6899} (\bibinfo {year} {2001})}\BibitemShut {NoStop}%
\bibitem [{\citenamefont {Holevo}(1998)}]{Holevo1998capacity}%
  \BibitemOpen
  \bibfield  {author} {\bibinfo {author} {\bibfnamefont {A.~S.}\ \bibnamefont
  {Holevo}},\ }\bibfield  {title} {\bibinfo {title} {The capacity of the
  quantum channel with general signal states},\ }\href
  {https://doi.org/10.1109/18.651037} {\bibfield  {journal} {\bibinfo
  {journal} {IEEE Transactions on Information Theory}\ }\textbf {\bibinfo
  {volume} {44}},\ \bibinfo {pages} {269} (\bibinfo {year} {1998})}\BibitemShut
  {NoStop}%
\bibitem [{\citenamefont {Schumacher}\ and\ \citenamefont
  {Westmoreland}(1997)}]{schumacher1997classical}%
  \BibitemOpen
  \bibfield  {author} {\bibinfo {author} {\bibfnamefont {B.}~\bibnamefont
  {Schumacher}}\ and\ \bibinfo {author} {\bibfnamefont {M.~D.}\ \bibnamefont
  {Westmoreland}},\ }\bibfield  {title} {\bibinfo {title} {Sending classical
  information via noisy quantum channels},\ }\href
  {https://doi.org/10.1103/PhysRevA.56.131} {\bibfield  {journal} {\bibinfo
  {journal} {Phys. Rev. A}\ }\textbf {\bibinfo {volume} {56}},\ \bibinfo
  {pages} {131} (\bibinfo {year} {1997})}\BibitemShut {NoStop}%
\bibitem [{\citenamefont {Shapourian}\ \emph {et~al.}(2021)\citenamefont
  {Shapourian}, \citenamefont {Liu}, \citenamefont {Kudler-Flam},\ and\
  \citenamefont {Vishwanath}}]{shapourian2021negativity}%
  \BibitemOpen
  \bibfield  {author} {\bibinfo {author} {\bibfnamefont {H.}~\bibnamefont
  {Shapourian}}, \bibinfo {author} {\bibfnamefont {S.}~\bibnamefont {Liu}},
  \bibinfo {author} {\bibfnamefont {J.}~\bibnamefont {Kudler-Flam}},\ and\
  \bibinfo {author} {\bibfnamefont {A.}~\bibnamefont {Vishwanath}},\ }\bibfield
   {title} {\bibinfo {title} {Entanglement negativity spectrum of random mixed
  states: A diagrammatic approach},\ }\href
  {https://doi.org/10.1103/PRXQuantum.2.030347} {\bibfield  {journal} {\bibinfo
   {journal} {PRX Quantum}\ }\textbf {\bibinfo {volume} {2}},\ \bibinfo {pages}
  {030347} (\bibinfo {year} {2021})}\BibitemShut {NoStop}%
\bibitem [{\citenamefont {Goldstein}\ \emph
  {et~al.}(2006{\natexlab{b}})\citenamefont {Goldstein}, \citenamefont
  {Lebowitz}, \citenamefont {Tumulka},\ and\ \citenamefont
  {Zangh\`{\i}}}]{goldstein2006canonical}%
  \BibitemOpen
  \bibfield  {author} {\bibinfo {author} {\bibfnamefont {S.}~\bibnamefont
  {Goldstein}}, \bibinfo {author} {\bibfnamefont {J.~L.}\ \bibnamefont
  {Lebowitz}}, \bibinfo {author} {\bibfnamefont {R.}~\bibnamefont {Tumulka}},\
  and\ \bibinfo {author} {\bibfnamefont {N.}~\bibnamefont {Zangh\`{\i}}},\
  }\bibfield  {title} {\bibinfo {title} {Canonical typicality},\ }\href
  {https://doi.org/10.1103/PhysRevLett.96.050403} {\bibfield  {journal}
  {\bibinfo  {journal} {Phys. Rev. Lett.}\ }\textbf {\bibinfo {volume} {96}},\
  \bibinfo {pages} {050403} (\bibinfo {year} {2006}{\natexlab{b}})}\BibitemShut
  {NoStop}%
\bibitem [{\citenamefont {Teufel}\ \emph {et~al.}(2024)\citenamefont {Teufel},
  \citenamefont {Tumulka},\ and\ \citenamefont
  {Vogel}}]{Teufel2024gencanonical}%
  \BibitemOpen
  \bibfield  {author} {\bibinfo {author} {\bibfnamefont {S.}~\bibnamefont
  {Teufel}}, \bibinfo {author} {\bibfnamefont {R.}~\bibnamefont {Tumulka}},\
  and\ \bibinfo {author} {\bibfnamefont {C.}~\bibnamefont {Vogel}},\ }\bibfield
   {title} {\bibinfo {title} {Canonical typicality for other ensembles than
  micro-canonical},\ }\href {https://doi.org/10.1007/s00023-024-01466-7}
  {\bibfield  {journal} {\bibinfo  {journal} {Annales Henri Poincar\'e}\
  }\textbf {\bibinfo {volume} {26}},\ \bibinfo {pages} {1477} (\bibinfo {year}
  {2024})}\BibitemShut {NoStop}%
\end{thebibliography}%

\end{document}